%% file: main.tex
\newif\iftpdp\tpdpfalse

\documentclass[11pt,letterpaper]{article}
\usepackage[margin=1in]{geometry}
\usepackage{times}
\usepackage[dvipsnames]{xcolor}
\usepackage{graphicx} 
\usepackage{amsmath, amsfonts, enumerate, amssymb, amsthm, nccmath}
\usepackage{nicefrac}

\usepackage{algorithmic}
\usepackage[linesnumbered,ruled,vlined]{algorithm2e}
\usepackage{typed-checklist}
\usepackage[shortlabels,inline]{enumitem}

\usepackage{thm-restate}
\usepackage{multirow}
\usepackage{booktabs}

\usepackage{DefaultPackages}
\usepackage{MathematicA}

\newcommand{\var}[1]{\Var\sbk{#1}}

\title{Parameter-Free Triangle Counting}

\author{Asaf Etgar\thanks{Department of Applied Mathematics, Yale University} \and Anna Gilbert\thanks{Department of Statistics and Data Science, Yale University} \and Quanquan Liu\thanks{Department of Computer Science, Yale University} \and Andrew McGregor\thanks{Manning College of Information \& Computer Sciences, University of Massachusetts Amherst}}
\date{}

\usepackage{mathtools}
\usepackage{hyperref}
\definecolor{mydarkblue}{rgb}{0,0.08,0.45}
\hypersetup{ %
    colorlinks=true,
    linkcolor=mydarkblue,
    citecolor=mydarkblue,
    filecolor=mydarkblue,
    urlcolor=mydarkblue,
}

\newtheorem{theorem}{Theorem}[section]
\newtheorem{claim}[theorem]{Claim}

\newtheorem{corollary}[theorem]{Corollary}

\newtheorem{remark}[theorem]{Remark}

\theoremstyle{definition}
\newtheorem{definition}[theorem]{Definition}

\newcommand{\eps}{\varepsilon}

\input{macro}

\newcommand{\cA}{\mathcal{A}}

\DeclareMathOperator{\Mult}{Mult}

\newcommand{\tguess}{\tau}
\newcommand{\fail}{\texttt{FAIL}}
\newcommand{\wt}{\widetilde}
\DeclareMathOperator{\oracle}{{oracle}}
\newcommand{\light}{\texttt{L}}
\newcommand{\heavy}{\texttt{H}}
\newcommand{\expec}[1]{\mathbb E\left [ #1 \right ]}
\newcommand{\hTh}{\widehat{T}_\heavy}
\newcommand{\hTl}{\widehat{T}_\light}
\newcommand{\logpass}{\textsc{LogPass}}
\newcommand{\warmup}{\textsc{OnePass}}
\newcommand{\loglogpass}{\textsc{LogLogPass}}
\newcommand{\ppass}{\textsc{p-Pass}}

\usepackage{tikz}
\usetikzlibrary{shapes.geometric, arrows.meta, positioning, fit, calc, backgrounds, decorations.pathreplacing, decorations.pathmorphing}

\usepackage[capitalize,nameinlink]{cleveref}
\crefname{algocf}{alg.}{algs.}
\Crefname{algocf}{Algorithm}{Algorithms}

\usepackage[colorinlistoftodos,textsize=small,color=red!25!white,obeyFinal]{todonotes}

\iftpdp
\newcommand{\qq}[1]{}
\else
\newcommand{\qq}[1]{{\color{purple} Quanquan: #1}}
\fi

\allowdisplaybreaks

\crefalias{AlgoLine}{line}
\crefname{algocfline}{Line}{Lines}

\newcommandx{\Asaf}[2][1=]{\todo[linecolor=blue,backgroundcolor=blue!25,bordercolor=blue,#1]{#2}}

\newcommandx{\Anna}[2][1=]{\todo[linecolor=green,backgroundcolor=green!25,bordercolor=green,#1]{#2}}

\newcommandx{\QQ}[2][1=]{\todo[linecolor=purple,backgroundcolor=purple!25,bordercolor=purple,#1]{#2}}

\newcommandx{\Andrew}[2][1=]{\todo[linecolor=yellow,backgroundcolor=yellow!25,bordercolor=yellow,#1]{#2}}

\declaretheorem[name=Lemma, numberwithin=section]{lemma}

\begin{document}

\begin{titlepage}
\maketitle

\input{abstract}

\thispagestyle{empty}
\end{titlepage}

\pagenumbering{gobble}
\newpage
\setcounter{tocdepth}{2}
\tableofcontents
\newpage
\cleardoublepage
\newpage
\pagenumbering{arabic}

\clearpage
\sloppy
\addtocontents{toc}{\protect\setcounter{tocdepth}{1}}

\input{introduction}
\input{background}
\input{logn_pass_algorithm}

\input{mixed-bound-formulation}

\input{lower-bound}


\bibliographystyle{alpha}
\bibliography{ref}

\appendix
\input{verified-algorithm-proof}

\end{document}

%% file: macro.tex
\DeclareMathOperator{\ccRR}{R}
\DeclareMathOperator{\cost}{cost}
\DeclareMathOperator{\err}{err}

\DeclareMathOperator{\ccout}{out}
\newcommand{\tritest}{\mbox{\sc TriTest}}
\newcommand{\toltritest}{\mbox{\sc TolTriTest}}

\DeclareMathOperator{\poly}{poly}

%% file: abstract.tex
\begin{abstract}
Given an undirected, unweighted graph, the triangle counting problem seeks the number of occurrences of three-cycles in a graph dataset. Triangle and, more generally, subgraph counting are classical and important problems in graph algorithms. These problems are central to a wide variety of applications, including community detection, computing the clustering coefficient, motif discovery in protein networks, and broad social network analysis. In many of these applications, the graph datasets are so voluminous or are collected via incremental updates that we model them as streams of updates to an underlying graph. It is in this model that we seek to count the number of triangles $T$ in a graph $G = (V,E)$ with $n$ vertices and $m$ edges. 

There are a number of foundational results for the streaming triangle counting problem, both theoretical in nature and developed for practical application. There is, however, one major drawback to \textit{all} previous constant space algorithms: to achieve both a constant factor approximation to the triangle count and to obtain the sublinear space guarantees of these algorithms, 
one needs to know \textit{a priori} a constant factor approximation of the triangle count. 
Thus, to obtain the best space guarantees for estimating the number of triangles, one must already have a 
good estimation of the number of triangles! This requirement is inherently circular.

We initiate the study of \textit{parameter-free streaming triangle counting}, without any \textit{a priori} 
knowledge of $T$ or any quantities depending on $T$, provided $m$, the length of the stream. We describe a family of $O(p)$ pass parameter-free triangle counting algorithms that guarantee a mixed multiplicative and additive approximation of $T$ and use $\widetilde{O}\rbk{\frac{m+T}{\sqrt{T}}}$ expected space. Moreover, this family leads to an $O(\log\log(n))$ pass algorithm that gives a $(1+\eps)$ multiplicative approximation of $T$ with the same space complexity. These algorithms rely on the notion of a \emph{verified} parametrized algorithm: An algorithm that is parametrized by $\tau$ and which provides either an approximation of $T$ when $\tau \le T$, or declares that $T < \tau$. Furthermore, we demonstrate a lower bound: any parameter-free algorithm that provides a multiplicative approximation for all values of $T$ must use $\Theta(m)$ space, even on streams where the actual triangle count is moderately large.

\end{abstract}

%% file: introduction.tex
\section{Introduction}

Given an undirected, unweighted graph $G = (V, E)$, the \emph{triangle counting problem} seeks
to determine the number of three-cycles in $G$. Triangle counting, and more generally subgraph
counting, are among the most classical and well-studied problems in graph algorithms, with a rich
body of foundational results in both the theoretical~\cite{atserias2013size, bera2017towards,
seshadhri2013triadic, kolountzakis2012efficient, bar2002reductions, kallaugher2019complexity,
jayaram2021optimal, mcgregor2020triangle, mcgregor2016better, chen2022triangle,
cormode2017second} and applied~\cite{tangwongsan2013parallel, pavan2013counting, jha2015space,
shin2020fast, jha2015counting, shin2018think, wang2019rept, yang2022distributed, shin2018dislr,
shin2021cocos, shin2017wrs, lee2020temporal} communities. The triangle count of a graph encodes
a wealth of structural information.

In particular, streaming triangle counting has applications
across a broad range of data science problems, including: community detection in social networks
(see e.g.~\cite{palla2005uncovering, prat2012shaping, ST21} and references therein); computing
the well-known clustering coefficient for broad applications like network
neuroscience~\cite{MSEW18}; finding structural holes and small worlds in
communities~\cite{LGFZT20}; recommendation systems~\cite{ZZLSZ14}; semantic user
search~\cite{LZZY18}; social and network science measurements~\cite{granovetter1977strength,
watts1998collective, newman2004finding, farkas2001spectra, foucault2010friend,
leskovec2008microscopic}; clustering~\cite{tseng2021parallel}; and motif discovery in protein
networks~\cite{milo2002network}. There is even an annual competition devoted to fast triangle
counting solutions~\cite{GraphChallenge}, and comprehensive surveys of the broader significance
of subgraph counting can be found in~\cite{al2018triangle} and~\cite{easley2010networks}.

In many of these applications, graph data arrive incrementally and cannot be stored in their
entirety, motivating the study of triangle counting in the \emph{streaming model}. As a concrete
example, consider large-scale recommendation systems, which must adapt to a continuously growing
item set. Such additions are naturally modeled as a stream of edge insertions, and in corporate
settings the underlying graphs can reach hundreds of billions of entries. For datasets of this
scale, one requires accurate triangle counts using space that is sublinear in the size of the
input. Accordingly, there is sustained and growing interest in approximate, constant-pass triangle
counting algorithms for edge-insertion streams, where the goal is to achieve the best possible
approximation using the smallest possible memory footprint.

Many prior algorithms achieve this goal via clever sampling: a randomly chosen subset of stream
edges is maintained, from which the triangle count is estimated and then scaled up. Unfortunately,
under worst-case ordering such techniques meet known space lower bounds: $\Omega(m)$ space can be
required for a stream of $m$ edges when the triangle count $T$ is
even moderately small~\cite{kallaugher2017hybrid,BravermanOV13}, and even in random-order streams
$\Omega\! (\varepsilon^{-2}\sqrt{m/T})$ space is unavoidable when $T \leq \sqrt{m}$ \cite{mcgregor2020triangle}. 

There is, however, a further and more fundamental
limitation shared by all prior low-space algorithms: to simultaneously achieve a
constant-factor approximation of $T$ and the sublinear $\widetilde{O}(m/\sqrt{T})$ space
guarantee, these algorithms require a priori knowledge of a constant-factor approximation of $T$
itself. While it is reasonable to assume knowledge of $m$, it is unrealistic to assume that we
have an approximation of $T$ in real-world networks, since two similarly sized graphs can have
vastly different triangle counts. This requirement is inherently circular: to allocate the right
amount of memory before the stream begins, one must already know, to within a constant factor,
the very quantity one is trying to compute. Allocating too much space wastes resources;
allocating too little causes the algorithm to miss the few triangles present and produce a wildly incorrect estimate.

One natural remedy is to run successive passes, decreasing a guess $T_i$ for a lower bound on
$T$ by a factor of two each time and restarting the algorithm whenever the current guess fails.
Setting the initial guess to $T_0 = m^{3/2}$, this strategy terminates after $O(\log n)$ passes.
Our goal is to do better: we seek an algorithm that requires no lower bound on $T$ or any
quantity depending on $T$ --- that is, a \emph{parameter-free} algorithm. We initiate this study
in the random-order streaming model, where edges arrive in the order of a uniformly random permutation, and
present a family of parameter-free algorithms that offer a tradeoff between number of passes and quality of approximation. We present both a single-pass algorithm and an $O(\log \log n)$ pass algorithm in this setting. Recall that assuming a random ordering is necessary to achieve sub-linear space even when a constant factor approximation to $T$ is known \cite{kallaugher2017hybrid,BravermanOV13}.

\subsection{Our Results and Technical Overview}

Our main result is an $O(p)$-pass parameter-free triangle counting algorithm in the random-order streaming model for any $p \leq \lceil\log \log(n)\rceil $. That is, we prove the following.

\begin{theorem}[Main Result, informal]
Let $S$ be a stream of edges (in random order) representing a simple graph $G$ with $n$ vertices, $m$ edges, and $T$ triangles. Then there exists an $O(p)$ pass algorithm that uses expected 
$\wt{O}\rbk{\delta^{-1}\varepsilon^{-3.5}\rbk{\frac{m + T}{\sqrt{T}}}}$ space, and with probability $1-\delta - 1/\poly(n)$ outputs $\widehat T$, an approximation of $T$, with 
\[
    T \leq \widehat T \leq (1+\varepsilon) T + 
    \widetilde{O}(n^{1/2^p}).
\]
The algorithm does not require any prior knowledge of $T$.
\end{theorem}
If $T < m$, the space bound simplifies to  $\wt{O}\rbk{\delta^{-1}\varepsilon^{-3.5}\rbk{m/\sqrt{T}}}$. When $p = \lceil\log \log(n)\rceil$, a simple modification gives a $1 + \epsilon$ approximation $\widehat T$ to $T$, with no additive error. 


\paragraph{Mixed approximation bound.}
We refer to the type of approximation result given as a \emph{mixed approximation} bound. More generally, we say that $X$ is an $(\alpha,\beta)$ approximation of $T$ if $T\le X \le \alpha T + \beta$. Our main result is, therefore, an $(1+\varepsilon,\widetilde{O}(n^{1/2^p}))$ mixed approximation of $T$ with an $O(p)$ pass parameter-free algorithm, and a $(1+\eps,0)$ approximation in $O(\log\log(n))$ passes.

\paragraph{Verified triangle counting.}
A key building block of our algorithm is a \emph{verified} triangle counting subroutine, which
we develop by modifying the algorithm of McGregor and Vorotnikova~\cite{mcgregor2020triangle}.
All prior triangle counting algorithms take a parameter $\tau$ intended to be a lower bound on
$T$ and produce an estimate without any guarantee of correctness if $\tau$ is not actually a
lower bound. Our verified subroutine, by contrast, either outputs a $(1 \pm \varepsilon)$
approximation of $T$ when $T \geq \tau$, or outputs \fail~when $T < \tau$ (with high
probability). This verification property is crucial: it allows us to search over possible values
of $\tau$ across passes without fear of silently returning an incorrect estimate, as each
invocation either certifies its own output or explicitly declares failure. We define this result in Section~\ref{sec:warmup} with additional details in Section~\ref{appendix:verifiedalgorithm}.

\paragraph{Stream segmentation and the one-pass result.}
A na\"{i}ve approach to parameter-free counting would geometrically decrease a guess $\tau_i$ by
a constant factor in each pass, starting from $m^{3/2}$. This requires $O(\log n)$ passes, and
reducing the pass count further by using a larger factor comes at a steep cost in space. To
circumvent this, we observe that a randomly ordered stream of $m$ edges can be partitioned into
$O(\log n)$ contiguous segments, each of which behaves statistically like an independent uniform
sample of the edges. Provided the maximum per-edge triangle participation $\Delta_e$ is not too
large, each segment is \emph{representative}: its triangle count is a $(1 \pm \varepsilon)$
fraction of the expected value. This allows us to run $O(\log n)$ geometrically spaced guesses
for $\tau$ across the segments of a \emph{single pass}, rather than across separate passes.
The result is a one-pass algorithm that either produces a $(1 \pm \varepsilon)$ approximation of
$T$, or certifies that $T = \widetilde{O}(n)$. We discuss this result in Section~\ref{subsec:stream segment} with more details in the subsequent section.

\paragraph{Heavy and light edges, $O(p)$ pass and  $O(\log \log n)$-pass algorithms.}
When designing sampling based \emph{parametrized} triangle counting algorithms, the parameter $\tau$ is used to separate the edges in the graph to \emph{heavy} and \emph{light}, where heavy roughly corresponds to edges participating in $\sqrt{\tau}$ triangles. This separation, and the assumption that $\tau$ is a good approximation of $T$, allow one to set sampling probabilities and approximate the triangles with or without heavy edges separately. We bypass the need for a parameter $\tau$ by \emph{adaptively} setting a heaviness threshold $\gamma$ according to the certificate provided by the verified algorithm in a previous pass. 
If the one-pass algorithm did not return an estimate, we have established that $T = \widetilde{O}(n)$, allowing a space budget of $\widetilde{O}(m/\sqrt{n})$ in subsequent passes. We use this budget to separate edges into \emph{heavy} edges,
those involved in at least $\gamma = \widetilde{\Theta}(\sqrt{n})$ triangles, and \emph{light}
edges, the rest. 
An oracle for this separation is constructed by sampling each vertex
independently with a small probability and collecting all incident edges; a Chernoff bound
argument shows that this oracle correctly classifies every edge with high probability. Triangles
involving at least one heavy edge are then counted directly from the sampled edge set, while
triangles consisting entirely of light edges are handled by running the one-pass subroutine on
the substream of light edges. If this round again fails to produce an estimate, we learn that $T
= \widetilde{O}(\sqrt{n})$, shrinking the threshold $\gamma$ accordingly and repeating. After
$i$ rounds we have either output an approximation of $T$ or deduced that $T =
\widetilde{O}(n^{1/2^i})$. After $p$ rounds, one can output a partial estimate and an upper bound of $\widetilde{O}(n^{1/2^p})$ on the remaining fraction, resulting in the $(1+\eps,\widetilde{O}(n^{1/2^p}))$ mixed approximation. After $O(\log \log n)$ rounds, keeping all edges of the stream is within the space budget, providing an exact count.
\paragraph{Lower bound.}
We complement our algorithmic results with a space lower bound showing that  parameter-free triangle estimation in a random-order stream is inherently expensive: even if the true number of triangles is moderately large, any algorithm that must return an $(\alpha,\beta)$-approximation without being told the relevant scale in advance can still be forced to use $\Omega(m/\sqrt{\beta})$ space, and in the fully parameter-free setting this becomes essentially linear space. In words, the algorithm cannot safely ``wait until it realizes'' that the graph has many triangles, because before it learns that, it must already have enough information to distinguish graphs with no triangles from graphs with only a small number of triangles, and that alone is hard.

The technical contribution is a communication-complexity reduction that strengthens the McGregor--Vorotnikova lower bound from ordinary triangle testing to a \emph{tolerant} version: by adjoining a gadget that guarantees many triangles while revealing almost nothing new to Alice, the proof shows that any protocol that works even on graphs with at least $\tau_2 \ge \tau_1 \ge 0$ triangles must still pay the full $\Omega(m/\sqrt{\tau_1})$ communication needed to distinguish $0$ from $\tau_1$ triangles. This tolerant communication lower bound is then transferred to the random-order streaming model, yielding the stated space lower bound for parameter-free triangle estimation.

%% file: background.tex
\section{Preliminaries}
\subsection{Definitions and Problem Statement}


Let $G = (V,E)$ be a simple graph with $|V| = n$ and $|E| = m$. We wish to count the number of triangles (three-cycles) in $G$, $T(G)$ or simply $T$. Of course, a naive algorithm is to iterate over all triples of vertices and check how many of them form triangles. However, when $G$ is very large, this
is often wasteful or even impossible when we cannot store the entirety of $G$ in memory. In the \emph{insert-only random order} graph streaming model, we see the edges of the graph one by one, arriving in the order of a uniformly random permutation on the edges. We refer to this stream as $S$. We wish to obtain an approximation of $T$ using sublinear (in $m$ the number of edges) space. The goal is to minimize the space used during the algorithm, as well as the number of passes over the stream. The success probability of the algorithm is dependent on the random permutation and the internal coin tosses of the algorithm.

As it turns out, our algorithms provide \emph{mixed approximation} bounds on $T$ of the form:
\begin{definition}(Mixed approximation)
    We say that $X$ is an $(\alpha,\beta)$ approximation of $T$ if $T\le X \le \alpha T + \beta$. If $\beta = 0$, we simply say that $X$ is an $\alpha$ approximation, or specify that it is a multiplicative approximation.
\end{definition}

A desirable mixed approximation result would therefore be of the form: fix $\varepsilon>0$ and return $X$, a $1+\varepsilon$ approximation of $T$ with high probability.


To fix notation and terminology, a key quantity in our analysis is the number of triangles adjacent to a given edge. For an edge $e$ and a set of edges $F$, let $t_e^F$ be the number of triangles in $\{e\}\cup F$ that include $e$. Let $t_e:=t_e^E$ where $E$ is the edge set of the graph. \emph{Light} edges have small values of $t_e$ while \emph{heavy} edges have large values of $t_e$ (where small and large are defined as necessary in a subsequent algorithm). We also have $\Delta_e(G) = \Delta_e = \max_{e} t_e$.

\subsection{Related results}
As we discussed in the introduction, all previous results require a lower bound estimate on the number of triangles in the entire stream. It is precisely this requirement that we eliminate.

In the multi‑pass edge-streaming model, where the algorithm may scan the input more than once, Cormode and Jowhari~\cite{cormode2014second} proved that any constant‑pass algorithm for arbitrary graphs requires $\Theta\left(\frac{m}{\sqrt{T}}\right)$ space, and they gave a two‑pass algorithm achieving this bound. Building on this result, Kolountzakis et al.~\cite{kolountzakis2012efficient} presented a three‑pass algorithm that uses $O\left(m^{1/2}+\frac{ m^{3/2}}{T}\right)$ space, while Bera and Chakrabarti~\cite{bera2017towards} later showed how to attain $O\left(\frac{m^{3/2}}{T}\right)$ space in just four passes.

In the adjacency‑list streaming model, where the neighbors of each vertex arrive in a contiguous block (so each edge is observed twice), McGregor, et al.~\cite{mcgregor2016better} gave a one‑pass algorithm using $O\left(\frac{m}{\sqrt{T}}\right)$ space. More recently, Kallaugher, et al.~\cite{kallaugher2019complexity} designed a two‑pass algorithm requiring $O\left(\frac{m}{T^{2/3}}\right)$ space, and they proved matching lower bounds (under standard communication‑complexity conjectures) for both one‑ and two‑pass algorithms in this model.

Finally, the query model considers the number of edge‑ or vertex‑queries rather than space usage. Although this setting is quite different, it faces the same challenge of ``heavy'' vertices or edges. Eden, et al.~\cite{eden2017approximately} initiated the study of triangle counting under query complexity bounds, and their techniques were extended by Eden et al.~\cite{eden2018approximating} to count general cliques and by Assadi et al.~\cite{AKK19} to count all constant‑size subgraphs.

%% file: logn_pass_algorithm.tex
\section{Warmup}
\label{sec:warmup}

\subsection{Verified Triangle Counting}
In many of the algorithms found in the literature, the assumption that the input parameter $\tau$ is a lower bound poses two complications: The first is the assumption of \emph{knowledge} of such a lower bound $\tau$. The second is the \emph{inability to verify} that $\tau$ is indeed a lower bound. That is, the algorithms will necessarily output some estimate, where in reality $\tau$ is not even a lower bound on $T$ - so the estimate need not be accurate. This raises the need for a \emph{verified} triangle counting algorithm: An algorithm that not only produces an estimate, but declares when the input parameter doesn't satisfy the assumptions placed on it. 
We modify a previous algorithm to provide such an algorithm and prove the following theorem:

\begin{restatable}{theorem}{verifiedalg}\label{thm:bb}
There exists a single-pass algorithm $\cA$ which takes one pass over a randomly ordered stream and parameter $\tguess$ such that  a) the algorithm uses $s(m,\tau) = \wt{O}(\epsilon^{-2}\delta^{-1} (m/\sqrt{\tguess} +T/\sqrt{\tau}) )$ space in expectation, and b) with probability at least $1-\delta-1/\poly(n)$: 
\begin{itemize}
\item If $T\geq \tguess$, the output is  a $(1\pm \epsilon)$-approximation of $T$.
\item If $T<\tguess$,  the output is a $(1\pm \epsilon)$-approximation of $T$ or  \fail.
\end{itemize}
Note that the algorithm has no a priori information about $T$.
\end{restatable} 
\begin{remark}
    To fix terminology, we say that an execution of $\Aa$ is \emph{successful} if it outputs a correct estimate or \fail. That is \ref{thm:bb} means $\prob{\Aa\text{ is successful}} \ge 1-\delta - 1/\poly(n)$.
\end{remark}

The algorithm is a modification of the algorithm described in \cite{mcgregor2020triangle}, and is not the main focus of our paper. We provide the details and proof in Appendix \ref{appendix:verifiedalgorithm} for completeness, and will mostly use the algorithm as a black box. However we use two aspects of the algorithm explicitly:
\begin{enumerate}[label=\textup{(\roman*)}]
    \item \label{property: algorithm test}The algorithm produces an estimate $X$, and tests whether $X \ge (1-\varepsilon)\tau$. If so, the output is $X$. Otherwise, the output is \fail.
    \item \label{property: algorithm upper bound} The output $X$ satisfies $X \le (1+\varepsilon) T + \varepsilon\tau$ with probability $1-\delta-1/\poly(n)$. 
\end{enumerate}

From this point on, we let $\Aa$ be the algorithm promised by Theorem~\ref{thm:bb} and its space complexity $s(m,\tau)$. 

\subsection{$O(\log n)$ pass algorithm}
With a verified algorithm at hand, devising a parameter free $O(\log(n))$ pass algorithm is simple: There are at most $m^{3/2}$ triangles in any graph, so execute $\Aa$ with $\tau = m^{3/2}$. If the estimate was too large, set $\tau_1 = \frac{1}{2}m^{3/2}$ and execute $\Aa$ again, with $\tau_1$ as a new guess of a lower bound. In each subsequent pass, we decrease our guess by a factor of $2$, set $\tau_j = 2^{-j}m^{3/2}$ and execute $\Aa$. Note that when $t = \log_2(m^{3/2}) = O(\log(n))$, $\tau_t = 1$, so there must be some $j$ where $2\tau_j \ge T \ge \tau_j$. Therefore we terminate at pass $j = O(\log(n))$. Formally, this is captured by the following claim.
\input{algorithms/logpass}
\begin{claim}
    With probability $1-\delta - 1/\poly(n)$ $\logpass$ returns a $(1\pm\varepsilon)$ estimate of $T$ using $O(s(m,T))$ expected space.
\end{claim}
\begin{proof}
    Note that $\tau_t = 1$, so let $j^*$ be the first time $\tau_{j^*} \le T$. Then $\Aa(\tau_{j^*})$ returns a $(1\pm \varepsilon)$ approximation of $T$. $\logpass$ does not return a correct estimate if and only if at pass $i\le {j^*}$, $\Aa(\tau_i)$ did not succeed. This happens with probability at most $\delta + 1/\poly(n)$. So 
\begin{fleqn}[\parindent]
    \begin{align*}
    \prob{\logpass \text{ outputs a correct estimate}} &=
    \prob{\bigcap_{j\le j^*} \Aa \text{ succeeds in pass $j$}} \\
    &\ge 1-j^*\delta - j^*/\poly(n) \\
    &\ge 1-O(\log(n)\delta) - 1/\poly(n).        
    \end{align*}
\end{fleqn}
    The result follows by rescaling $\delta$.

    Since $s(m,\tau)$ is a decreasing function of $\tau$, and $\tau_j$ is a decreasing function of $j$, $s(m,\tau_j)$ is an increasing function of $j$. Since the algorithm terminates at some pass $i$ with $i\le j^*$, the expected space usage of $\logpass$ is upper bounded by $s(m,\tau_{j^*}) = \wt{O}\rbk{\varepsilon^{-2}\delta^{-1}\rbk{\frac{m + T}{\sqrt{\tau_{j^*}}}}} = \wt{O}\rbk{\varepsilon^{-2}\delta^{-1}\rbk{\frac{m + T}{\sqrt{T}}}}$ where the last equality is since $\tau_{j^*} \le T \le 2\tau_{j^*}$.
\end{proof}

The choice to decrease our guess by a factor of $2$ with each pass is quite arbitrary, and results in the logarithmic number of passes. One could ask what would happen if we were to use a larger factor decrease $r$ in order to get at most $c = O(1)$ passes? That is, we look for $r$ such that we can set $\tau_j = r^{-j}m^{3/2}$ and $r^{-c}m^{3/2} = 1$. Consequently, $r = m^{3/2c}$. The correctness of the algorithm with geometrically decreasing guesses with a factor of $r$ goes through just the same, however there is a cost in terms of space complexity: Suppose that $T \le m$. Then $s(m,\tau) = \wt{O}(m/\sqrt{\tau})$. On pass $j$ where $\tau_j \le T \le r\tau_j$ we terminate, and use expected space $s(m,\tau_j) \le s(m,T/r) = \wt{O}(m/\sqrt{T/r}) = \wt{O}(m^{1+3/4c}/\sqrt{T})$. That is, the reduced number of passes comes at a cost of space blowup; We wish to maintain a small space complexity - even at the cost of more than constant passes, or a worse approximation.

\subsection{Segmenting a Stream}
\label{subsec:stream segment}
The aforementioned space blowup illustrates the need to search for a different approach.  A first attempt is to partition a prefix of a randomly-ordered stream into $1/p$ segments $\langle S_1 \mid S_2 \mid \cdots \mid S_{1/p}\rangle $ where each edge $e\in S_i$ with probability $p$. Deciding on a \emph{fixed} size for each segment $m_i$ makes the analysis difficult, as each segment is a sample of $m_i$ edges from $G$ \emph{without replacement}. To circumvent this issue, we sample the set of $1/p$ lengths from a \emph{multinomial} distribution. This solution is captured by the following lemma:
\begin{lemma} Let $S$ be a length $m$ random-order stream of the edge set $E$. Sample $(m_1, \ldots, m_{1/p})\sim \Mult(m;p, \ldots, p)$, and let $\langle S_1 \mid S_2 \mid \cdots \mid S_{1/p}\rangle$ be a segmentation of $S$ into $1/p$ segments, where $|S_j| = m_j$. Then for each $F\subset E$ with $|F| = k$, $\prob{S_i = F} = p^k(1-p)^{m-k}$.
    
\end{lemma}
\begin{proof}
 The probability $S_i=F$ is
\[
\prob{m_i=k}\cdot  \prob{\text{random subset of }k \text{ elements equals } F}=\binom{m}{k} p^k (1-p)^{m-k} \frac{1}{\binom{m}{k}} = p^k (1-p)^{m-k}.
\]
\end{proof}
That is, the distribution of $S_i$ is the same as a sampled set of edges from $E$ where each edge is added to $S_i$ independently with probability $p$.
We use this argument repeatedly. Therefore, while the segments are simply contiguous parts of a random order stream, we refer to them as sets obtained by independently sampling edges with the appropriate probability. We face two problems:
(i) Different segments of the stream may not be representative of the entirety of the stream, and (ii) We need to ensure that we advance to a sufficiently accurate guess before we run out of data.

The first challenge is less of a concern for randomly ordered streams, and is handled by the following lemma:
\begin{lemma}[Lemma \ref{lem:conc}, informal] Let $G$ be a graph with $n$ vertices and $T$ triangles, and let $T_p$ be the number of triangles in the graph formed by sampling each edge independently with probability $p$. Then with probability at least $1-1/n^4$, $T_p$ is within $\pm\varepsilon p^3\max(\tau,T)$ of $p^3T$, where $\tau = \wt{\Omega}(\varepsilon^{-2}p^{-1}\eta + \varepsilon^{-2}p^{-5})$ and $\eta \ge \Delta_e$.
    
\end{lemma}
That is, as long as we have an upper bound on $\Delta_e$, the number of triangles in each segment is sufficiently concentrated around its mean. There is a natural upper bound of $n$ on $\Delta_e$ which leads us to a preliminary result: A one pass algorithm that either outputs an approximation of $T$, or declares that $T$ is upper bounded by $\wt{O}(n)$.

\input{algorithms/onepass}

Note that $t$ is the smallest integer satisfying $2^{-t}m^{3/2} \le \tau^*$, so $2^{-(t-1)}m^{3/2} > \wt{\Omega}(\varepsilon^{-2}n) \ge 1$, therefore $2^{t-1} \le m^{3/2}$ which implies $t = O(\log(n))$. 

To fix notation, for each $j\in [t]$, we call $S_j$ a \emph{segment} of $S$, and let $T_j$ be the number of triangles in the $j$th segment. We say that a segment is \emph{representative} if $T_j\in (1\pm\varepsilon)t^{-3}\max(T,\tau^*)$. In particular, when $T \ge \tau^*$, a representative segment satisfies $T_j\in(1\pm\varepsilon)t^{-3}T$.

\begin{claim} \label{claim:warmup algorithm}
    If $\warmup$ terminates and outputs a numerical value, then with probability at least $1-\delta - 1/\poly(n)$ it is a $(1\pm\varepsilon)^2$ approximation of $T$. 
\end{claim}
\begin{remark}
    Note that when $\varepsilon$ is small enough, $(1\pm\varepsilon)^2 \approx (1\pm 2\varepsilon)$, so rescaling $\varepsilon$ results in the desired approximation.
\end{remark}

\begin{proof}
    Since we assume $\warmup$ terminated with a numerical value it must terminate when $\Aa(S_j,(1-\varepsilon)\tau_j t^{-3})$ outputs some estimate of $T_j$. Let that estimate be $X$. We need to show that $X \in (1\pm\varepsilon)^2t^{-3}T$. There are two cases to check: The termination happened when $\tau_j$ was a lower bound on $T$, or $\tau_j$ was not a lower bound on $T$, yet $\Aa$ returned an estimate. 


    Suppose $T \ge \tau_j$ when $\warmup$ terminates. By minimality of $t$, $\tau_j \ge \tau_t \ge \tau^*/2$, so $T\ge \tau^*/2$ and $\max(T,\tau^*)\le 2T$. Hence, if $S_j$ is representative, $T_j \ge (1-2\varepsilon)t^{-3}T \ge (1-2\varepsilon)t^{-3}\tau_j$. That is, $(1-2\varepsilon)t^{-3}\tau_j$ is a lower bound on $T_j$. Hence the output $X$ satisfies $$X\in (1\pm\varepsilon)T_j \subseteq (1\pm\varepsilon)(1\pm2\varepsilon) t^{-3}T$$ as claimed.

    Alternatively, suppose $T < \tau_j$ and $\warmup$ still terminated. Then $\Aa$ outputs some $X$, even though $(1-\varepsilon)\tau_j t^{-3}$ is not necessarily a lower bound on $T_j$. We claim that $X$ is still an approximation of $T_j$. We show that by demonstrating that $\tau_j$ can't be far from a lower bound on $T$, and therefore not far from a lower bound on $T_j$: Suppose by way of contradiction that $T < \tau_j/2$. Since $\Aa(S_j)$ produced an output $X$, by property \ref{property: algorithm test} it must satisfy
    \[
    (1-\varepsilon)(1-\varepsilon)t^{-3}\tau_j \le X.
    \]
    However, 
    \begin{align*}
        X &\le (1+\varepsilon)T_j + \varepsilon(1-\varepsilon)\tau_jt^{-3} &\text{By property \ref{property: algorithm upper bound}}\\
        & \le (1+\varepsilon)t^{-3}\sbk{T + \varepsilon\max(T,\tau^*)} + \varepsilon t^{-3}\tau_j &\text{By Lemma \ref{lem:conc}}\\
        & \le (1+\varepsilon)t^{-3}\sbk{\frac{\tau_j}{2} + \varepsilon \tau_j + \varepsilon\tau_j}& \text{By $T\le \tau_j/2$ and $\tau^*\le \tau_j$}\\
        & \le (1+\varepsilon)\rbk{\frac{1}{2} + 2\varepsilon} t{^-3}\tau_j \\
        & \le \frac{3}{4}\rbk{1+\frac{1}{2}\varepsilon}t^{-3}\tau_j
    \end{align*}
    where the last inequality holds when $\varepsilon < 1/10$.
    Therefore
    \[
    (1-\varepsilon)(1-\varepsilon)t^{-3}\tau_j \le X \le \frac{3}{4}\rbk{1+\frac{1}{2}\varepsilon}t^{-3}\tau_j
    \]
    Contradiction when $\varepsilon < 1/10$. Therefore, $\tau_j/2 \le T \le \tau_j$, and the output $X$ satisfies
    \[
     (1-\varepsilon)^2t^{-3}T \le (1-\varepsilon)^2t^{-3}\tau_j \le X \le (1+\varepsilon)t^{-3}[T + \varepsilon\max(T,\tau^*)] + 2\varepsilon t^{-3}T \le (1+\varepsilon)(1+4\varepsilon)t^{-3}T.
    \]
    Rescaling $\varepsilon$ gets the desired result.

    We now analyze the probability of success. Let $\Rr$ be the event that all segments are representative. By the union bound over all segments and \ref{lem:conc}, $\prob{\Rr}  \ge 1-t/n^4 \ge 1-O(\log(n)/n^4)\ge 1-1/n^3$. Let $\Mm$ be the event that the output of $\warmup$ is a $(1\pm\varepsilon)$ approximation of $T$. For each $j$, let $\Ee_{\Aa,j}$ be the event that $\Aa$ succeeds on segment $S_j$. Conditioned on $\Rr$, if $\Ee_{\Aa,i}$ hold for all $i$, then $\warmup$ succeeds. That is,
    \[
    \prob{\Mm\mid \Rr} \ge \prob{\bigcap_{i\le t}\Aa(S_i)\text{ succeeds}\mid \Rr} \ge (1-t\delta - t/\poly(n))\ge(1-O(\log(n)\delta)) - 1/\poly(n)).
    \]
    We conclude
    \[
    \prob{\Mm}\ge \prob{\Mm\mid \Rr}\cdot\prob{\Rr} \ge(1-O(\log(n)\delta)) - 1/\poly(n))(1-n^{-3}).
    \]
    Rescaling $\delta$ provides the desired result.

\end{proof}
\begin{corollary}\label{cor:warmup did not terminate}
    If \warmup~returns \fail, then with probability $1-\delta - 1/\poly(n)$, $T = \wt{O}(n)$. 
\end{corollary}
\begin{proof}
    Suppose $T \ge \tau^*$. Then there exists some $\tau_j$ such that $T \ge \tau_j$. By \ref{claim:warmup algorithm}, $\warmup$ terminates and outputs an approximation of $T$. Equivalently, if $\warmup$ did not output an estimate, then $T < \tau^* = \wt{O}(n)$.
\end{proof}

%% file: algorithms/logpass.tex
\begin{figure}[!htb]
\begin{center}
\fbox{
\begin{minipage}{0.95\textwidth}
{
\noindent {\logpass}\\
\textbf{Input:} A random order stream $S$ with $|S| = m$, a verified algorithm $\Aa$.
\begin{itemize}
    \item Let $t = \lceil\log(m^{3/2})\rceil = O(\log(n))$.
    \item At pass $i\in [t]$, set $\tau_i = 2^{-i}m^{3/2}$ and run $\Aa(S,\tau_i)$. If it returns $\fail$, continue to the next pass. Otherwise, return its output as the estimate.
    \item If all passes return $\fail$, output $0$.
\end{itemize}
}
\end{minipage}
}
\end{center}
\caption{An $O(\log(n))$-pass parameter free triangle counting algorithm}
\end{figure}

%% file: algorithms/onepass.tex
\begin{figure}[!htb]
\begin{center}
\fbox{
\begin{minipage}{0.95\textwidth}
{
\noindent {\warmup}\\
\textbf{Input:} A random order stream $S$ with $|S| = m$.
\begin{enumerate}[start=0]
    \item {\em Preprocessing: }
\begin{itemize}
\item Let $t$ be the smallest integer such that $2^{-t}m^{3/2} \le c\epsilon^{-2} (n \ln(n)t +\ln^3(n)t^5)$ and set $\tau^* = c\epsilon^{-2} (n \ln(n)t +\ln^3(n)t^5)$.
\item Sample $(m_1, \ldots, m_t)\sim \Mult(m;1/t, \ldots, 1/t)$ .
\item Segment the length-$m$ stream into $\langle S_1 | S_2 | \ldots | S_{t}  \rangle$ where $|S_j|=m_j$ for $j\in [t]$. 
\end{itemize}
\item {\em One Pass:} For each $j\in [t]$, run $\Aa(S_j,(1-\varepsilon)t^{-3}\tau_j)$ where $\tau_j = 2^{-j}m^{1.5}$. Output its estimate scaled by $t^3$ if the output is not \fail. Otherwise, continue to the next segment. 
\end{enumerate}
If none of the outputs were an estimate, return $\fail$.
}
\end{minipage}
}
\end{center}
\caption{A one-pass parameter free triangle counting algorithm that either produces a $(1\pm\varepsilon)$ approximation of $T$, or guarantees that $T = \wt{O}(n)$.}  \label{fig:warmuponepass}
\end{figure}

%% file: mixed-bound-formulation.tex
\section{$p$ Pass Algorithm}

In this section, we prove the main theorem below.

\begin{theorem}[$p$-pass Mixed Approximation Algorithm]\label{them:loglogn algorithm}
    Let $S$ be a stream of edges representing a simple graph $G$ with $n$ vertices, $m$ edges and $T$ triangles. Then there exists an $O(p)$ pass algorithm that uses expected $\wt{O}\rbk{\delta^{-1}\varepsilon^{-3.5}\rbk{\frac{m + T}{\sqrt{T}}}}$ space, and with probability $1-\delta - 1/\poly(n)$ outputs a mixed $(1+\varepsilon, \widetilde{O}(n^{1/2^p}))$ approximation of $T$. The algorithm does not require any prior knowledge of $T$.
\end{theorem}
\begin{remark}
    If $T < m$, the space bound simplifies to $\wt{O}\rbk{\delta^{-1}\varepsilon^{-3.5}\rbk{m/\sqrt{T}}}$.
\end{remark}
\begin{remark}
    In fact, our algorithm may terminate in fewer than $O(p)$ passes, in which case we are guaranteed a $1+\varepsilon$ multiplicative approximation.
\end{remark}

Before formalizing the algorithm, we briefly outline the intuition for the algorithm and the analysis. At a high level, our algorithm relies on dynamically changing our space budget based on estimation results for smaller space used earlier in the stream: the failure of earlier estimation subroutines using smaller space directly leads to a larger space budget for subsequent passes. Recall that to achieve the target space bound of $\widetilde{O}(m/\sqrt{T})$, an algorithm is permitted to use more space precisely when $T$ is small. When the initial one-pass algorithm fails to produce an estimate, we gain crucial structural information: we establish that $T = \widetilde{O}(n)$. This upper bound guarantees a minimum space budget of $\widetilde{O}(m/\sqrt{n})$ for the subsequent pass. We spend this budget to isolate heavy edges, which are edges that are involved in at least $\gamma_1 = \widetilde{\Theta}(\sqrt{n})$ triangles, and count their incident triangles directly.

This leaves the set of light edges in the stream where the maximum per-edge triangle count (i.e. number of triangles incident to one edge) $\Delta_e$ is strictly bounded by $\gamma_1$. We then execute the one-pass subroutine \emph{exclusively} on these light edges, while maintaining an estimate of the triangles involving at least one heavy edge. If this round also fails to produce an estimate, we deduce that the remaining light triangle count is at most $\widetilde{O}(\sqrt{n})$, and that these triangles are a relatively large fraction of the total number of triangles. Consequently, we deduce a new upper bound on $T$ of $\widetilde{O}(\sqrt{n})$. This tighter upper bound recursively allows for an even larger space budget of $\widetilde{O}(m/n^{1/4})$, allowing us to shrink the heavy-edge threshold for the next round to $\gamma_2 = \widetilde{\Theta}(n^{1/4})$, and so on. In general, after $i$ failed rounds, the upper bound on the remaining uncounted triangles is reduced from $\widetilde{O}(n^{1/2^{i-1}})$ to $\widetilde{O}(n^{1/2^i})$. Because this upper bound shrinks by a square root in each iteration, the exponent halves, requiring only $O(\log \log n)$ rounds for the threshold to collapse to a polylogarithmic quantity. Once this occurs, keeping all remaining edges of the stream falls entirely within our permitted space budget, allowing us to compute the remaining triangles exactly and output a $(1+\epsilon)$ multiplicative approximation.



\input{algorithms/loglogpass_mixed}

\subsection{Correctness}
The proof of correctness is organized the following way: We first provide two probabilistic arguments that we use repeatedly: (i) concentration guarantees for the number of triangles in a given segment of the stream, and (ii) a success probability for distinguishing heavy and light edges. We then proceed to prove that when \ppass~outputs a numerical value during round $i\le p$, it is a multiplicative approximation of $T$. Finally, we show that the final pass outputs the correct mixed approximation of $T$, and conclude that the output of \ppass~is correct. 

\begin{remark}\label{remark: multiplicative to mixed one-sided}
    Throughout our analysis, we assume $\varepsilon < 1/10$. Note that in that case, $(1+\varepsilon)(1-\varepsilon)^{-1} \le(1+\varepsilon)( 1+2\varepsilon) = 1+4\varepsilon$. Therefore, if $X$ satisfies $(1-\varepsilon)T\le X \le (1+\varepsilon)T$, it also satisfies
$T \le (1-\varepsilon)^{-1}X\le (1+4\varepsilon)T$. Thus, proving that $X$ is a two-sided multiplicative approximation of $T$ implies that $(1-\varepsilon)^{-1}X$ is a $(1+4\varepsilon,0)$ mixed one-sided approximation of $T$. We omit these calculations in most lemmas, and remember that rescaling $\varepsilon$ results in the desired approximation.
\end{remark}

\begin{restatable}[Concentration of Triangles]{lemma}{lemconc}\label{lem:conc} Given a simple graph on $n$ nodes with $T$ triangles, let $T_p$ be the number of triangles in the graph formed by sampling each edge in $G$ independently with probability $p$.
Then $\EE[T_p]=p^3 T$. 
Furthermore, 
\[ \Pr[|T_p-p^3 T|\geq \epsilon p^3\max(\tau,T) ]\leq 1/n^4 \  \] 
where $\tau=c(\epsilon^{-2} p^{-1}\eta \ln n+\epsilon ^{-2} p^{-5} \ln^3 n)$, $\eta \ge \Delta_e $, and $c>0$ is some sufficiently large constant.
\end{restatable}

\begin{proof}
We may restrict our attention to the $m'\leq 3T$ edges that are involved in triangles. Label these edges $e_1,\ldots, e_{m'}$ and let $t_i$ be the number of triangles that include $e_i$. Let $X_i=1$ if the $e_i$  is sampled and $X_i=0$ otherwise. Let $Y_i$ be the number of length two paths between the endpoints of $e$ that are sampled. Let $f(X_1,\ldots,X_{m'})$ be the number of triangles sampled. Let $\Gamma$ be the event that for  $Y_i\leq c_i:=2 p^2 t_{i}+30 \ln n$ for all $i\in [m']$. Note that for all $x,x'\in \{0,1\}^m$ that only differ in position $i$,
\[
|f(x)-f(\tilde{x})|\leq \begin{cases}
c_i & x\in \Gamma \\
t_i & \mbox{otherwise} 
\end{cases}
\] 
Then setting $\gamma_i= \min(1,c_i/(t_i-c_i)$) in \cite[Theorem 4]{Warnke2016TypicalBoundedDifferences} establishes that 
\begin{eqnarray*}
& & \Pr[f(X_1,\ldots, X_{m'})\geq p^3 T+\epsilon p^3 \max( T, \tau )] \\
& \leq &   
\exp\left ( -
\frac{\epsilon^2 p^6 \max( T, \tau )^2}{2p(1-p)\sum_{i\in [m']}(2c_i)^2+2/3 \cdot \epsilon  p^3 \max( T, \tau )\max_i (2c_i)}
\right) + (1-\Pr[\Gamma])\sum_i \frac{1}{\gamma_i} 
\end{eqnarray*}
To bound the second term, first note that by an application of the Chernoff bound
\[
\Pr[Y_i>c_i]=\Pr[Y_i>\EE[Y_i]+(1+\frac{30\ln n}{p^2 t_i})\EE[Y_i]]
\leq \begin{cases}
 \exp(-p^2 t_i/3) & \mbox{ if $p^2 t_i\geq 30\ln n$}\\ 
 \exp(-(30 \ln n)/3) & \mbox{ if $p^2 t_i\leq 30\ln n$}
\end{cases}
\]
and is at most $n^{-10}$ in either case. Hence, we can bound the second term using the  union bound as follows:
\[
(1-\Pr[\Gamma])\sum_{i\in [m']} \frac{1}{\gamma_i}
\leq 3T \sum_{i\in [m']} n^{-10} \leq n^{-6}
\]

To bound the first term, note that by using the inequality that $(x+y)^2 \leq 2x^2+2y^2$,
\begin{eqnarray*}
 & & 2p(1-p)\sum_{i\in [m']}(2c_i)^2+2/3 \cdot \epsilon  p^3 T \max_i (2c_i)\\ 
& \leq & 8 p\sum_{i\in [m']}(2 p^4 t_{i}^2+1800 \ln^2 n)+4/3 \cdot \epsilon  p^3 \max( T, \tau ) \max_i c_i\\ 
&= & 
16p\sum_{i\in [m']}p^4 t_{i}^2+14400 p\sum_{i\in [m']}  \ln^2 n+4/3 \cdot \epsilon p^3 \max( T, \tau)\max c_i\\
&\leq & 
48 p^5 T \eta  + 43200 pT \ln^2 n+4/3 \cdot  \epsilon  p^3 \max( T, \tau ) \max c_i   
\\
&\leq & 
\max( T, \tau ) \cdot (48 p^5 \eta  + 43280 p \ln^2 n+8/3 \cdot  \epsilon p^5 \eta)   \ 
\end{eqnarray*}
Hence, 
\begin{eqnarray*}
\Pr[f(X_1,\ldots, X_{m'})\geq (1+\epsilon)p^3 T] &\leq & \exp \left (
- \frac{\epsilon^2 p^6 \max( T, \tau^* )}{\Theta(p^5 \eta  + p\ln^2 n+ \epsilon  p^5 \eta )}
\right )+n^{-6}
\end{eqnarray*}
The lower tail bound follows identically.
\end{proof}
\begin{remark}
    In particular, when $T \ge \tau$, with probability $\ge 1-n^{-4}$, $T_p\in (1\pm\varepsilon)p^3T$.
\end{remark}
We proceed to show the guarantees of the oracle used in pass $3i$.
\begin{restatable}{lemma}{orcguarantees}\label{lem:oracleguarantees}
    Let $1\le \gamma \le n$ be some predetermined threshold. Let $V_p$ be a collection of vertices sampled independently with probability $p = \min(1,50\varepsilon^{-2}(\log n)/\gamma)$. Let $O$ be the set of edges incident to at least one vertex in $V_p$. Define a function
    \[
    \oracle(e)=
    \begin{cases}
    \light & \mbox{ if } t_e^{O}< p \gamma\\
    \heavy & \mbox{ otherwise}
    \end{cases}
    \]
    
    With probability $(1-1/n^3)$, if $\oracle(e) = \heavy$ then $t_e^O = (1\pm\varepsilon)t_e p$, and if $\oracle(e)=\light$ then $t_e\le 2\gamma$.
\end{restatable}
\begin{proof}
    For every edge, $t_e^O\sim \Bin{t_e,p}$. The result follows from standard applications of the Chernoff bound and taking the union bound over all edges. We provide a detailed proof in \ref{appendix:verifiedalgorithm} for completeness.
\end{proof}

We use this lemma repeatedly, with different thresholds $\gamma$ and different sampling probabilities. We mention that the correctness of the first pass is proved in Claim~\ref{claim:warmup algorithm} and proceed to the main part of the algorithm.
\subsubsection{Correctness of termination at round $i$.}
For a given $i\le p$, the passes $3i-1,3i,3i+1$ are collectively called a \emph{round}. The successful execution of round $i$ depends on three random events: 
\begin{enumerate}
    \item The probability that the oracle constructed correctly distinguished \heavy \, and \light \, edges.
    \item \emph{Condition on the oracle's success}, the probability that the segments of $S^\light$ are representative of $S^\light$, the light edges in the stream. 
    \item Conditioned on the above two events, the probability that the execution of $\Aa$ on a segment $S^\light_j$ succeeds. 
\end{enumerate}

To make the second event precise, note that the classification of edges is a function of the sampled vertex set and the edge set alone, and in particular is independent of the ordering of the stream. Hence, conditioned on the oracle, $S^\light$ is a uniformly random-order stream of the $m_\light$ light edges. Let $T^\light_j$ denote the number of triangles in $S^\light_j$. We say that $S^\light_j$ is \emph{representative} if
\[
|T^\light_j - t^{-3}T_\light| \le \varepsilon t^{-3}\max(T_\light, U_i).
\]
On the success event of Lemma~\ref{lem:oracleguarantees}, every light edge satisfies $t_e\le 2n^{1/2^i}$, so Lemma~\ref{lem:conc} applies to the light subgraph with $\eta = 2n^{1/2^i}$ and threshold $U_i$, and each segment is representative with probability at least $1-1/n^4$.
The probabilistic analyses of the next lemmas all follow from the same line of reasoning, captured by the following:
\begin{lemma}[Successful executions of $\Aa$ in round $i$]\label{lem: loglog round i probability of blackbox success}
    Let $\Mm_i$ be the event that in round $i$, the execution of $\Aa$ on all segments of $S^\light$ was successful. Then 
    $\prob{\Mm_i} \ge 1-\delta - 1/\poly(n)$.
\end{lemma}
\begin{proof}
    
For a given $i$, let $\Ee_{i,oracle}$ be the event that \ref{lem:oracleguarantees} holds at round $i$, and let $\Ee_{i,j,rep}$ be the event that the segment $S_j^\light$ is representative of $S^\light$ in round $i$. Let $\Ee_{i,rep} = \cap_{j\in [t]}\Ee_{i,j,rep}$. By the union bound and \ref{lem:conc},
\[
\prob{\Ee_{i,rep}\mid \Ee_{i,\oracle}} \ge 1-1/n^3
\]

Let $\Mm_{i,j}$ be the event that in round $i$ the execution of $\Aa$ on the segment $S^\light_j$ was successful, and let $\Mm_{i} = \cap_{i\in [t]}\Mm_{i,j}$. By the union bound, 
\[
\prob{\Mm_i \mid \Ee_{i,rep},\Ee_{i,\oracle}}\prob{\bigcap_{j\in [t]}\Mm_{i,j}\mid \Ee_{i,rep},\Ee_{\oracle}}\ge 1-t\delta -1/\poly(n)  = 1-O(\log(n)\delta)-1/\poly(n).
\]
Consequently, 
\begin{align*}
    \prob{\Mm_{i}}&\ge \prob{\Mm_{i}\mid \Ee_{i,rep},\Ee_{i,\oracle}}\cdot \prob{\Ee_{i,rep},\Ee_{i,\oracle}} \\
    &=\prob{\Mm_{i}\mid \Ee_{i,rep},\Ee_{i,\oracle}}\cdot \prob{\Ee_{i,rep}\mid \Ee_{i,\oracle}}\cdot\prob{\Ee_{i,\oracle}} \\
    &\ge (1-O(\log (n)\delta) - 1/\poly(n))(1-1/n^3)(1-1/n^2) \\
    &= 1-O(\log(n)\delta) - 1/\poly(n).
\end{align*}
Rescaling $\delta$ shows the desired result.
\end{proof}

\begin{lemma}[Heavy Triangles Estimation]
\label{lem:heavytrianglesest}
With probability $1-1/n^3$, at round $i$, $\hTh \in (1\pm\varepsilon)T_\heavy $.
\end{lemma}
\begin{proof}
    By \ref{lem:oracleguarantees} with the threshold $\gamma = n^{1/2^i}$, with probability $1-1/n^3$, $t_e^O = (1\pm\varepsilon)t_ep_i$. We then note that 
    \[\hTh = p_i^{-1}\left(\sum_{e : \oracle(e) = \heavy} t_{e,1}^O + t_{e,2}^O/2 + t_{e,3}^O/3\right)
    \]and the result follows. The coefficients of $t^O_{e,i}$ compensate for overcounting - a triangle that has $i$ heavy edges attached to it is counted $i$ times.
\end{proof}

\begin{lemma}[Termination with $t^{3}\hTl + \hTh$]\label{lem: loglog terminiation output H + L} Suppose that during round $i$ $\ppass$ terminates and outputs $t^{3}\hTh + \hTl$. Then with probability at least $1-\delta - 1/\poly(n)$ it is a  $(1\pm \varepsilon)$ approximation of $T$.
\end{lemma}
\begin{proof}
    Since $T = T_\heavy + T_\light$, we only need to show that $t^3 \hTl \in (1\pm\varepsilon) T_\light$. The proof is similar in structure to the proof of \ref{claim:warmup algorithm}, however since we look over the stream $S^\light$ and not $S$, the specifics change slightly.

    Let $\tau_j$ be the guess when $\Aa(S_j^\light,(1-\varepsilon)\tau_j t^{-3})$ outputs an estimate, and recall that that estimate is denoted $\hTl$. There are two cases to check: The termination happened when $\tau_j$ was a lower bound on $T_\light$, or $\tau_j$ was not a lower bound on $T_\light$, yet $\Aa$ returned an estimate.

Suppose $T_\light \ge \tau_j$ when round $i$ terminates. Since every guess satisfies $\tau_j \ge \tau^* \ge U_i$ by definition, we have $T_\light \ge \tau^*$, so $\max(T_\light,U_i) = T_\light$ and representativeness gives $T^\light_j \ge (1-\varepsilon)t^{-3}T_\light \ge (1-\varepsilon)t^{-3}\tau_j$. That is, $(1-\varepsilon)t^{-3}\tau_j$ is a lower bound on $T^\light_j$. By the guarantees of $\Aa$, 
$$\hTl\in (1\pm\varepsilon)T^\light_j \subseteq (1\pm\varepsilon)^2 t^{-3}T_\light,
\qquad\text{hence}\qquad t^{3}\hTl\in (1\pm\varepsilon)^2 T_\light$$
as claimed.

    Alternatively, suppose $T_\light < \tau_j$ yet $\Aa(S_j^\light, (1-\varepsilon)\tau_j t^{-3})$ returned an estimate $X$. We claim that $X$ is still an approximation of $T_j$. We show that by demonstrating that $\tau_j$ can't be far from a lower bound on $T_\light$, and therefore not far from a lower bound on $T^\light_j$: Suppose by way of contradiction that $T_\light < \tau_j/2$. Since $\Aa(S_j)$ produced an output $X$, by property \ref{property: algorithm test} it must satisfy
    \[
    (1-\varepsilon)(1-\varepsilon)t^{-3}\tau_j \le X.
    \]
    However, 
    \begin{align*}
        X &\le (1+\varepsilon)T^\light_j + \varepsilon(1-\varepsilon)\tau_jt^{-3} &\text{By property \ref{property: algorithm upper bound}}\\
        & \le (1+\varepsilon)t^{-3}\sbk{T_\light + \varepsilon\max(T_\light,U_i)} + \varepsilon t^{-3}\tau_j &\text{By Lemma \ref{lem:conc} over the graph induced by light edges}\\
        & \le (1+\varepsilon)t^{-3}\sbk{\frac{\tau_j}{2} + \varepsilon \tau_j + \varepsilon\tau_j}& \text{By $T_\light\le \tau_j/2$ and $U_i\le\tau^*\le \tau_j$}\\
        & \le (1+\varepsilon)\rbk{\frac{1}{2} + 2\varepsilon} t{^-3}\tau_j \\
        & \le \frac{3}{4}\rbk{1+\frac{1}{2}\varepsilon}t^{-3}\tau_j
    \end{align*}
    where the last inequality holds when $\varepsilon < 1/10$.
    Therefore
    \[
    (1-\varepsilon)(1-\varepsilon)t^{-3}\tau_j \le X \le \frac{3}{4}\rbk{1+\frac{1}{2}\varepsilon}t^{-3}\tau_j
    \]
    Contradiction when $\varepsilon < 1/10$. Therefore, $\tau_j/2 \le T_\light \le \tau_j$, and the output $X$ satisfies
    \[
     (1-\varepsilon)^2t^{-3}T_\light \le (1-\varepsilon)^2t^{-3}\tau_j \le X \le (1+\varepsilon)t^{-3}[T_\light + \varepsilon\max(T_\light,U_i)] + 2\varepsilon t^{-3}T_\light \le (1+\varepsilon)(1+4\varepsilon)t^{-3}T_\light.
    \]
    Rescaling $\varepsilon$ gets the desired result.
    The probability analysis is shown in \ref{lem: loglog round i probability of blackbox success}.
\end{proof}

\begin{corollary}\label{cor:TL theta tauj}
Suppose that during round $i$, $\ppass$ terminates and outputs $t^{3}\hTl + \hTh$, and let $\tau_j$ be the guess at which $\Aa$ returned an estimate. Then, on the success event of Lemma~\ref{lem: loglog round i probability of blackbox success}, $\tau_j/2 \le T_\light \le 2\tau_j$.
\end{corollary}
\begin{proof}
If $T_\light < \tau_j$, the analysis of the second case in Lemma~\ref{lem: loglog terminiation output H + L} shows $\tau_j/2\le T_\light \le \tau_j$. Otherwise $T_\light\ge\tau_j$, and $j$ is the first such index: at the first index $j_0$ with $\tau_{j_0}\le T_\light$, the input $(1-\varepsilon)t^{-3}\tau_{j_0}$ is a lower bound on $T^\light_{j_0}$, so $\Aa$ outputs an estimate there. We claim $\tau_{j-1}\le 2\tau_j$: if $\tau_j = 2^{-j}m_\light^{3/2}$ then $\tau_{j-1} = 2\tau_j$, and if $\tau_j = \tau^*$ then $2^{-j}m_\light^{3/2}\le \tau^*$, so $\tau_{j-1} = \max(2^{-(j-1)}m_\light^{3/2},\tau^*)\le 2\tau^* = 2\tau_j$. Hence $\tau_j \le T_\light < \tau_{j-1} \le 2\tau_j$.
\end{proof}
\begin{corollary}\label{cor: loglog L is small}
    Suppose that in round $i$, $\ppass$ did not output $t^{3}\hTl + \hTh$. Then with probability at least $1-\delta - 1/\poly(n)$, $T_\light < \tau^*$. Moreover, if \ppass~ did not terminate, $T_\light = \widetilde{O}(\varepsilon^{-2}n^{1/2^i})$.
\end{corollary}
\begin{proof}

Suppose $T_\light \ge \tau^*$. By definition $\tau_t = \tau^*\le T_\light$, so some executed guess satisfies $\tau_j \le T_\light$. On the success event of Lemma~\ref{lem: loglog round i probability of blackbox success}, at the first such $j$ the input $(1-\varepsilon)t^{-3}\tau_j$ is a lower bound on $T^\light_j$, so $\Aa$ outputs an estimate and the algorithm outputs $t^{3}\hTl + \hTh$. The contrapositive implies the first part of the corollary. If \ppass~ did not terminate, then in particular it did not output $\hTh$ at the end of the stream, so $T_\light < \tau^* = U_i = \widetilde{O}(\varepsilon^{-2}n^{1/2^i})$.
\end{proof}
\begin{lemma}[Termination with $\hTh$]\label{lem: loglog termination output H}
    Suppose $\ppass$ terminates at round $i$ and outputs $\hTh$. Then with probability $1-\delta - 1/\poly(n)$ it is a $(1\pm \varepsilon)$ approximation of $T$.
\end{lemma}
\begin{proof}
    Suppose $\ppass$ terminates at round $i$ and outputs $\hTh$. By the algorithm's design, $\tau^* = \varepsilon\hTh$, and by \ref{cor: loglog L is small}, $T_\light < \varepsilon\hTh$. Note that by \ref{lem:heavytrianglesest},
    $\varepsilon \hTh \le \varepsilon(1+\varepsilon)T_\heavy \le 2\varepsilon T_\heavy$. Therefore
    \[
    T = T_\heavy + T_\light \le T_\heavy + 2\varepsilon T_\heavy \le T_\heavy +2\varepsilon T,
    \]
    hence
    \[
    (1-2\varepsilon)T \le T_\heavy \le (1+2\varepsilon)T.
    \]
    That is, the output $\hTh$ is an approximation of $T_\heavy$, which in turn is an approximation of $T$, proving the claim. The probabilistic analysis is shown in \ref{lem: loglog round i probability of blackbox success}.
\end{proof}

We are ready to prove the correctness of $\ppass$. 
\begin{theorem}[Correctness of $\ppass$.]\label{thm:loglog correctness}
    With probability $1-\delta - 1/\poly(n)$, $\ppass$ outputs a $(1+4\varepsilon, \widetilde{O}\rbk{\varepsilon^{-2}n^{1/2^p}})$ one-sided approximation of $T$. 
\end{theorem}
\begin{proof}
    If the algorithm terminated during some round $i$, then by \ref{lem: loglog termination output H} and \ref{lem: loglog terminiation output H + L}, the output is a $(1\pm\varepsilon)$ approximation of $T$, and the result follows from \ref{remark: multiplicative to mixed one-sided}. Otherwise, the output is $(1-\varepsilon)^{-1}(\hTh + U_p)$. Since \ppass~ did not terminate during round $p$, by \ref{cor: loglog L is small} $T_\light \le U_p$. Therefore
    \[
    \hTh + U_p \ge (1-\varepsilon)T_\heavy + (1-\varepsilon)T_\light \ge (1-\varepsilon)T
    \]
    showing the lower bound. As for the upper bound, $\hTh \le (1+\varepsilon)T_\heavy$, therefore
    \begin{align*}
        (1-\varepsilon)^{-1}(\hTh + U_p) \le \frac{1+\varepsilon}{1-\varepsilon}T_\heavy + O(U_p) \le (1+4\varepsilon) T + O(U_p).
    \end{align*}
    The result follows by recalling that $U_p = \widetilde{O}(\varepsilon^{-2}n^{1/2^p})$.
\end{proof}

\subsection{Space complexity}
The space complexity of $\ppass$ is shown in two steps. First, we show that if $\ppass$ finished round $i$ with no output, then $T$ must be bounded from above, and these bounds are decreasing with $i$. We then use this bound to show that if the algorithm terminated at round $i$, then space used during the round does not exceed $\widetilde{O}(\varepsilon^{-3.5}m/\sqrt{T})$. Our analysis implies that the space complexity must increase with each round. Since we execute passes sequentially, the maximum space complexity is bounded by the maximum space complexity of all iterations.

\begin{lemma}[Non-termination]\label{lem: loglog nontermination}
    At the end of round $i$, if \ppass~did not terminate, then $T = \widetilde{O}(\varepsilon^{-3}n^{1/2^i})$ with probability $1-\delta - 1/\poly(n)$.
\end{lemma}
\begin{proof}
    Suppose $\ppass$ did not terminate at round $i$, then by \ref{cor: loglog L is small}, $T_\light,\varepsilon\hTh < \tau^* = U_i$. Since $T_\heavy/2 \le (1-\varepsilon)T_\heavy \le \hTh$, we conclude
    \[
    T = T_\light + T_\heavy \le \widetilde{O}(\varepsilon^{-2}n^{1/2^i}) + \widetilde{O}(\varepsilon^{-3}n^{1/2^i}) = \widetilde{O}(\varepsilon^{-3}n^{1/2^i}).
    \]
    The probability guarantees follow from \ref{cor: loglog L is small}
\end{proof}

\begin{lemma}[Space Complexity of Round $i$]\label{lem:loglog space round i}
    Suppose the algorithm terminates during round $i$, then the expected space used is $\widetilde{O}(\delta^{-1}\varepsilon^{-3.5}(m+T)/\sqrt{T})$.
\end{lemma}
\begin{proof}
    Let $\Ss_{\oracle}(i)$ be the space used by the oracle in round $i$, and $\Ss_{\Aa}(i,j)$ be the space complexity used by $\Aa$ with guess $\tau_j$ in round $i$. The expected space complexity of round $i$ is therefore $\Ss_{\oracle}(i) + \Ss_\Aa(i,j^*)$ where $j^*$ is the point where $\Aa$ terminates.
    Recall that  $s(m,\tau) = \widetilde{O}(\delta^{-1}\varepsilon^{-2}(m + T)/\sqrt{\tau})$, which is a decreasing function of $\tau$. Since the guesses $\tau_j$ are decreasing with $j$, $s(m,\tau)$ is an increasing function of $j$, so the space complexity of round $i$ is dominated by the point in which the round terminates. 
    
    Suppose first that $\Aa$ outputs an estimate, and that $\tau_j$ is the guess when the algorithm terminated. By Corollary~\ref{cor:TL theta tauj}, $T_\light = \Theta(\tau_j)$. Moreover, this implies $T_\light = \Omega(\tau_j) = \Omega(\tau^*) = \Omega(\varepsilon\hTh) = \Omega(\varepsilon T_\heavy)$ and therefore
    $\varepsilon T = \varepsilon T_\light + \varepsilon T_\heavy = O(T_\light)$, and as always $T = \Omega(T_\light)$. Consequently, 
    \begin{align*}
        \Ss_{\Aa}(i,j) &= \widetilde{O}\rbk{\delta^{-1}\varepsilon^{-2}\rbk{\frac{m_\light + T_\light}{\sqrt{(1-\varepsilon)\tau_j t^{-3}}}}} = \widetilde{O}\rbk{\delta^{-1}\varepsilon^{-2}\rbk{\frac{m + T}{\sqrt{T_\light}}}} \\
        &= \widetilde{O}\rbk{\delta^{-1}\varepsilon^{-2}\rbk{\frac{m + T}{\sqrt{\varepsilon T}}}} = \wt{O}\rbk{\delta^{-1}\varepsilon^{-2.5}\frac{m+T}{\sqrt{T}}}.
    \end{align*}

    Suppose now that the algorithm terminated and outputs $\hTh$. By the analysis in Lemma~\ref{lem: loglog termination output H}, $T_\light \le 2\varepsilon T_\heavy$, and consequently $T = \Theta(T_\heavy)$. Every execution of $\Aa$ in this round uses a guess $\tau_j\ge\tau^* = \varepsilon\hTh$, so the space of each is at most
    \[
    \Ss_{\Aa}(i,t) = \widetilde{O}\rbk{\delta^{-1}\varepsilon^{-2}\rbk{\frac{m_\light + T_\light}{\sqrt{(1-\varepsilon)\tau^* t^{-3}}}}} = \wt{O}\rbk{\delta^{-1}\varepsilon^{-2}\rbk{\frac{m+T}{\sqrt{\varepsilon \hTh}}}} = \wt{O}\rbk{\delta^{-1}\varepsilon^{-2.5}\rbk{\frac{m + T}{\sqrt{T}}}}.
    \]
    That is, $\Ss_\Aa(i,j^*) = \wt{O}\rbk{\varepsilon^{-2.5}\delta^{-1}\rbk{\frac{m+T}{\sqrt{T}}}}$.
    For $\Ss_{\oracle}(i)$, we only need to maintain $|O|$.
    Since the algorithm arrived at round $i$, it did not terminate at round $i-1$. By Lemma~\ref{lem: loglog nontermination} $\varepsilon T = \wt{O}(\varepsilon^{-2}n^{1/2^{i-1}})$, that is $\varepsilon^{3.5}\sqrt{T} = \wt{O}(\varepsilon^{-2}n^{1/2^i})$. We compute
    \[
    \Exp{|O|} = \sum_{v\in V} p_i \deg(v) = \widetilde{O}\rbk{\frac{m}{\varepsilon^2 n^{1/2^i}}} = \widetilde{O}\rbk{\frac{\varepsilon^{-3.5}m}{\sqrt{T}}}.
    \]

    We conclude 
    \[
    \Ss_{\oracle}(i) + \Ss_\Aa(i,j^*) = \wt{O}\rbk{\delta^{-1}\varepsilon^{-2.5}\rbk{\frac{m+T}{\sqrt{T}}}} +\wt{O}\rbk{\varepsilon^{-3.5}\rbk{\frac{m}{\sqrt{T}}}} = \wt{O}\rbk{\delta^{-1}\varepsilon^{-3.5}\rbk{\frac{m+T}{\sqrt{T}}}}.
    \]
\end{proof}

\begin{lemma}[Space Complexity] \label{lem:loglog space complexity}
    The expected space usage of \ppass~is $\widetilde{O}\rbk{\delta^{-1}\varepsilon^{-3.5}\rbk{\frac{m+T}{\sqrt{T}}}}$. 
\end{lemma}
\begin{proof}
    If the algorithm terminates at the first pass, the result follows from \ref{claim:warmup algorithm}.
    
    Otherwise, note that $s(m,\tau) = \wt{O}(\delta^{-1}\varepsilon^{-2}(m+T)/\sqrt{\tau})$ is a decreasing function of $\tau$. Consequently, for each $i$ $\Ss_\Aa(i,j)$ is an increasing function of $j$ and $\Ss_\Aa(i+1,\cdot) \le \Ss_\Aa(i,\cdot)$. That is, with each round we increase the space complexity used by the algorithm, so the total space complexity is dominated by the point in which $\ppass$ terminated.\\
    If $\ppass$ terminates at round $i$, then the claim follows from \ref{lem:loglog space round i}.
    If $\ppass$ did not terminate in all rounds $i\in [p]$, then by \ref{lem: loglog nontermination} $T = \wt{O}(\varepsilon^{-3}n^{1/2^p})$. The peak space usage therefore was at round $p$, where, by the same analysis as in \ref{lem:loglog space round i}, the space complexity is \[\wt{O}\rbk{\delta^{-1}\varepsilon^{-2}\frac{m+T}{\sqrt{(1-\varepsilon)\varepsilon^{-2}n^{1/2^p}}}} = \wt{O}\rbk{\delta^{-1}\varepsilon^{-3.5}\frac{m+T}{\sqrt{T}}}.\]
\end{proof}

\begin{proof}
    [Proof of Theorem \ref{them:loglogn algorithm}]
    Correctness follows from \ref{thm:loglog correctness}. The space complexity follows from \ref{lem:loglog space complexity}
\end{proof}

\subsection{A multiplicative approximation algorithm}
We designed a family of algorithms that allow for a tradeoff between the accuracy of approximation and the number of passes over the stream: $p$ passes guarantee a $(1+\varepsilon, \widetilde{O}(\varepsilon^{-2}n^{1/2^p}))$ approximation of $T$. Setting $p = \log\log(n)$, we can modify \ppass~ slightly and strengthen the guarantees.  Let \loglogpass~ be the following algorithm: Run \ppass~ with $p= \lceil\log\log(n)\rceil$ where the last pass simply keeps all edges of $S$ and counts triangles exactly.

\begin{theorem}[$\log\log(n)$ Pass Multiplicative Error]
\loglogpass~ returns a $(1+\varepsilon,0)$ mixed approximation of $T$ with probability at least $1-\delta - 1/\poly(n)$, using $\wt{O}(\delta^{-1}\varepsilon^{-3.5}(m+T)/\sqrt{T})$ expected space.
\end{theorem}
\begin{proof}
    By \ref{lem: loglog termination output H} and \ref{lem: loglog terminiation output H + L}, if \loglogpass~ terminates before the last pass, the output is a $(1+\varepsilon,0)$ approximation of $T$. In the last pass, we count exactly the number of triangles, so it's clearly the desired approximation. We only need to justify the space usage of the last pass: Since round $p = \log\log(n)$ did not end in termination, by \ref{lem: loglog nontermination}, $T = O(\varepsilon^{-1}U_p) = \widetilde{O}(\varepsilon^{-3})$. Therefore
    \[
    m = \wt{O}\rbk{\frac{m}{\sqrt{\varepsilon^{-3}\polylog(n)}}} = \wt{O}\rbk{\varepsilon^{-3.5}\frac{m}{\sqrt{T}}}.
    \]
\end{proof}

%% file: algorithms/loglogpass_mixed.tex
\begin{figure}[!htb]
\begin{center}
\fbox{
\begin{minipage}{0.95\textwidth}
{
\noindent {\ppass}\\
\textbf{Input:} A random order stream $S$ with $|S| = m$, desired number of passes $p$.
\begin{enumerate}[start=0]
    \item {\em Preprocessing: }
\begin{itemize}
\item Let $t$ be the smallest integer such that $2^{-t}m^{3/2} \le c\epsilon^{-2} (n \ln(n)t +\ln^3(n)t^5)$ and set $\tau^* = c\epsilon^{-2} (n \ln(n)t +\ln^3(n)t^5)$.
\item Sample $(m_1, \ldots, m_t)\sim \Mult(m;1/t, \ldots, 1/t)$.
\item Segment the length-$m$ stream into $\langle S_1 | S_2 | \ldots | S_{t}  \rangle$ where $|S_i|=m_i$ for $i\in [t]$. 
\end{itemize}
\item {\em First Pass:} For each $j\in [t]$, run $\Aa(S_j,(1-\varepsilon)t^{-3}\tau_j)$ where $\tau_j = 2^{-j}m^{1.5}$. Output its estimate scaled by $t^3$ if the output is not \fail. Otherwise, continue to the next segment.
\item For $i = 1,2,\ldots ,p$
\begin{itemize}
    \item {\em Pass $3i-1$:} Sample a subset of nodes where each node is sampled with probability $p_i=\min(1,50\epsilon^{-2} ( \log n) /n^{1/2^i})$. Collect the set of incident edges to sampled nodes $O$ and define a function \[
    \oracle(e)=
    \begin{cases}
    \light & \mbox{ if } t_e^{O}< p_i n^{1/2^i}\\
    \heavy & \mbox{ otherwise}
    \end{cases}
    \]
    \item {\em Pass $3i$:} Set $\hTh = 0, m_\light = 0$ For each edge $e$, if $\oracle(e) = \light$, increase $m_\light$ by $1$. Otherwise, $\hTh\gets \hTh + p_i^{-1}(t_{e,1}^O + t_{e,2}^O/2 + t_{e,3}^O/3)$. 

    \item {\em Pass $3i+1$:} Let $S^\light$ be the light edges in the stream.
    \begin{itemize}
        \item Let $t$ be the smallest integer such that $2^{-t}m^{3/2} \le c\epsilon^{-2} (n^{1/2^i} \ln(n)t +\ln^3(n)t^5)$. 
        \item Let $U_t = c\epsilon^{-2} (n^{1/2^i} \ln(n)t +\ln^3(n)t^5).$
        \item Set $\tau^* = \max(U_t,\varepsilon \hTh)$.
        \item Sample $(m_1, \ldots, m_t)\sim \Mult(m_\light;1/t, \ldots, 1/t)$.
        \item  Segment $S^\light$ into $\langle S^\light_1 | S^\light_2 | \ldots | S^\light_t\rangle$ where $|S^\light_j| = m_j$ for $j\in [t]$.
        \item For each $j\in [t]$ run $\Aa(S_j^\light,\tau_j(1-\varepsilon)t^{-3})$ where $\tau_j = \max\left(2^{-j}m_\light^{3/2},\tau^*\right)$. If $\Aa$ returns an estimate $\hTl$, output $t^3\hTl + \hTh$.
        \item If we reach the end of the stream and $\tau^* = \varepsilon\hTh$, output $\hTh$. Otherwise, mark this round as failed and continue to the next round.
    \end{itemize} 
\end{itemize}
\item \textit{Last Pass:} If we finish round $p$ with \fail, output $(1-\varepsilon)^{-1}(\hTh + U_p)$.
\end{enumerate}
}
\end{minipage}
}
\end{center}
\caption{An $O(p)$ pass parameter-free triangle counting algorithm. } \label{fig:p-pass}
\end{figure}

%% file: lower-bound.tex
\section{Lower Bound on Parameter Free Triangle Estimation}

In this section, we prove a space lower bound for our problem using the following formulation of the Random Partition Model. To prove our lower bound, we first consider a special case of the \emph{robust communication model} introduced by Chakrabarti, Cormode, and McGregor~\cite{ChakrabartiC016}.

\begin{definition}[Random Partition Communication Model]
Consider the following  two-player communication game for evaluating a function on an input graph $G=(V,E)$: 
\begin{enumerate}
\item  Let $\sigma: E\rightarrow \{A,B\}$ be a random mapping where $\Pr[\sigma(e)=A]=p$ for each $e\in E$ and  all edges are mapped independently. Let $\sigma_A(E)=\{e\in E:\sigma(e)=A\}$ and $\sigma_B(E)=\{e\in E:\sigma(e)=B\}$
\item Alice receives input $S_A= \sigma_A(E)$.
\item Bob receives input $S_B=\sigma_B(E)$, and 
\item Both have access to a string of public random bits $R$. 
\end{enumerate}
A one-way protocol $P$ to evaluate a function $f$ on $G$ has two components:
\begin{enumerate}
\item A message $M(S_A,R)\in \{0,1\}^*$ that Alice sends to Bob.
\item An output computed by Bob, $\ccout(M(S_A,R),S_B,R)$.
\end{enumerate}
We define the error of $P$ to be \[\err(P,f) ~:=~ \max_G \{\Pr[\ccout(M(S_A,R),S_B,R) \ne f(G)]\}\]
 where the probability is taken over the random choices of $\sigma$ and $R$.   Let 
 \[\cost(P) := \max_{G,\sigma, R} |M(S_A,R)|\]
 i.e., the longest message sent under any input, partition, and value of the public random string.
 Let 
 \[
 \cost_{\mathcal K} (P):= \max_{G\in \mathcal K,\sigma, R} |M(S_A,R)|\]
  be the cost of the protocol amongst graphs $G$ in some set $\mathcal K$. 
Define,
 \[
  \ccRR_{\delta,p}(f) ~:=~ \min_P \{\cost(P):\, \err(P,f) \le \delta\}
 \]
and let 
 \[
  \ccRR^{\mathcal K}_{\delta,p}(f) ~:=~ \min_P\{\cost_{\mathcal K}(P):\, \err(P,f) \le \delta\}
 \]
We emphasize that $ \ccRR^{\mathcal K}_{\delta,p}(f)$ is the maximum length of a message sent when $G\in \mathcal K$ by a protocol that has error at most $\delta$ on \emph{any} $G$.
 \end{definition}

We first define two communication problems.
\begin{definition}[Triangle Testing and Tolerant Triangle Testing]
For $0<\tau$, define the function $\tritest_{\tau,m}(\cdot)$ on graphs with at most $m$ edges and either $0$ or $\tau$ triangles where
\[
\tritest_{\tau,m}(G)= 
\begin{cases}
0 & \mbox{ if $G$ is triangle-free} \\
1 & \mbox{ if $G$ contains exactly $\tau$ triangles} \end{cases}
\]
For $0<\tau_1<\tau_2$, define the function $\toltritest_{\tau_1,\tau_2}(\cdot)$ on graphs at most $m$ edges with either $0$, $\tau_1$, or at least $\tau_2$ triangles where
\[
\toltritest_{\tau_1,\tau_2,m}(G)= 
\begin{cases}
0 & \mbox{ if $G$ is triangle-free} \\
1 & \mbox{ if $G$ contains exactly $\tau_1$ triangles} \\
2 & \mbox{ if $G$ contains at least $\tau_2$ triangles} 
\end{cases}
\]
\end{definition}
 
 We will need the following prior result about the communication complexity of $\tritest_{\tau,m}$.
 
\begin{theorem}[McGregor and Vorotnikova \cite{mcgregor2020triangle}]\label{thm:mv}
Let $m \in \mathbb N$ be sufficiently large and suppose:
\[0<\tau< \sqrt{m} \qquad  \qquad p=\tau^{-0.5}/10 \qquad  \qquad  \delta=1/m^4 \ .\]
Then $
 \ccRR_{\delta,p}(\tritest_{\tau,m})=\Omega(m/\sqrt{\tau})$. 
\end{theorem}

Note that in the hard instances used to establish the above bound, the number of nodes is $O(\sqrt{m}\tau)$ and therefore  the graph is sparse when $\tau$ is large.

Our main communication complexity result is proven in the next lemma. Essentially it establishes that even if there are at least $\tau_2$ triangles in the input graph, any protocol for $\toltritest_{\tau_1,\tau_2,m}$ must use sufficient communication to distinguish between $0$ and $\tau_1$ triangles because the first player may have no information where there were $\tau_2$ triangles.

\begin{lemma}[Tolerant Triangle Testing]\label{tttthm}
Let $m \in \mathbb N$ be sufficiently large and suppose:
\[0<\tau_1 < \sqrt{m/2} \qquad \qquad \tau_1< \tau_2 < m/4 \qquad \qquad p=\tau_1^{-0.5}/10 \qquad \qquad  \delta=1/m^4\]
Then $
 \ccRR^{\mathcal K}_{\delta,p}(\toltritest_{\tau_1,\tau_2,m})=\Omega(m/\sqrt{\tau_1})$ where $\mathcal K$  is the set of graphs with at most $m$ edges and at least $\tau_2$ triangles.
\end{lemma}

\begin{proof}
Let $P$ be a protocol for $\toltritest_{\tau_1,\tau_2,m}$ with error probability at most $\delta$. 
Suppose Alice and Bob want to evaluate  $\tritest_{\tau_1,m/2}$ on a graph $G_1$  with at most $m/2$ edges, where each edge is given to Alice independently with probability $p$ and to Bob otherwise.  They can use $P$ to solve $\tritest_{\tau_1,m/2}$ in the following way:

\begin{enumerate}
\item Define a gadget graph $J$ where $V(J)=\{u_1,\ldots,u_{\tau_2},t_1,t_2\}$ 
and let $E(J)$ contain all edges $\{u_i,t_j\}$ for $i\in[\tau_2]$ and $j\in\{1,2\}$. Let $H=G_1 \sqcup J$ be the disjoint union of $G_1$ and $J$.
\item Use public randomness to distribute the edges of $J$ between themselves such that each edge in $J$ is received by Alice with probability $p$.
\item Run $P$ to compute $\toltritest_{\tau_1,\tau_2,m}(H)$ with probability at least $1-\delta$. Note that since $J$ is triangle free, $\toltritest_{\tau_1,\tau_2,m}(H)=\tritest_{\tau_1,m/2}(G_1)$.
\end{enumerate}

The above reduction defines a protocol $P'$ for $\tritest_{\tau_1,m/2}$ with error at most $\delta$. Hence,
\[
\cost(P') \ge  \ccRR_{\delta,p}(\tritest_{\tau_1,m/2})
= \Omega(m/\sqrt{\tau_1})
\]
by Theorem~\ref{thm:mv}. Therefore, there exists a graph $G^*$  with at most $m/2$ edges, a partition $\sigma^*$ of the edges of $H^*=G^*\sqcup J$, and a setting $r^*$ of the public randomness such that Alice's message in protocol $P'$ has length 
\begin{equation} \label{eq:1}
 |M(S_A^*,r^*)|=\Omega(m/\sqrt{\tau_1}) 
\qquad \mbox{  where } \qquad S_A^*=\{e\in E(H^*): \sigma^*(E(H^*))=A\} \ .
 \end{equation} 

Now consider the graph  $H^+=G^*\sqcup J^+$, where $J^+$ is obtained from $J$ by adding the edge $e^*=\{t_1,t_2\}$. Note that $H^+$ has at least $\tau_2$ triangles. To bound $\cost_{\mathcal K} (P) $ we proceed as follows:
\begin{eqnarray*}
\cost_{\mathcal K} (P) 
 =   \max_{G\in \mathcal K,\sigma:E(G)\rightarrow \{A,B\}, R} |M(\sigma_A(E(G)),R)|  & \geq &  \max_{G=H^+,\sigma=\sigma^{+}, R=r^*} |M(\sigma_A(E(G)),R)|  \\
&  \geq &  |M(\sigma^+_A(E(H^+)),r^*)|  \\ 
& = & |M(S_A^*,r^*)|
\end{eqnarray*}
where $\sigma^{+}$ is the partition that extends $\sigma^*$ to a partition of $E(H^+)$ by assigning $e^*$ to Bob and assigning every other edge exactly as in $\sigma$. Note that with this definition $\sigma_A^+(E(H^+))=S_A^*$ as required in the last step. Hence, combined Equation \ref{eq:1}, we conclude  $\cost_{\mathcal K}(P)\ge \Omega(m/\sqrt{\tau_1})$ as required.
\end{proof}

We conclude with the following lower bound on parameter-free triangle estimation in the random order data stream model.

\begin{theorem}[Data Stream Lower Bound]
Let $\alpha>1$ and $\beta>0$ satisfy $\beta<\sqrt{m/2}$ and $2(\alpha+1)\beta < m/4$. Let $\mathcal A$ be a one-pass, parameter-free algorithm that returns an $(\alpha,\beta)$-approximation of $T$ with failure probability $1/m^4$ when the stream is in random order. Then $\mathcal A$ may require $\Omega(m/\sqrt{ \beta})$ space even when run on an input for any $T<m/4$.
\end{theorem}
\begin{proof}
We will use $\mathcal A$ to establish a communication complexity protocol for $\toltritest_{\tau_1,\tau_2}$ when $\tau_1=2 \beta$ and $\tau_2>2(\alpha+1) \beta$. Specifically, Alice runs $\mathcal A$ on a random permutation of her edges and then sends the memory state to Bob who instantiates $\mathcal A$ with the communicated memory state and continues running the algorithm on a random permutation of his edges. Note that the overall random permutation of the stream is random because Alice received each edge independently with probability $p$.  We next argue that even though $\mathcal A$ only returns an $(\alpha,\beta)$ approximation, this is sufficient to determine the value of $\toltritest_{\tau_1, \tau_2}$ on the input graph. Hence, by appealing to Theorem~\ref{tttthm}, the algorithm must use $\Omega(m/\sqrt{\tau_1})$ bits of memory even when there are $\tau_2$ triangles.

 If $T=0$, any $(\alpha,\beta)$ approximation is at most $\beta$. If $T_1=\tau_1$, any $(\alpha,\beta)$ approximation is at least $\tau_1$ and at most $\tau_1 \alpha+\beta$. If $T_2\geq \tau_2$ then, any $(\alpha,\beta)$ approximation is at least $\tau_2$.
Hence,
if 
\[ \beta < \tau_1 \quad \mbox{ and } \quad
\alpha \tau_1+\beta < \tau_2\]
then an $(\alpha,\beta)$ approximation is sufficient to distinguish between graphs with $0$ triangles, $\tau_1$ triangles, and at least $\tau_2$ triangles. This is the case since $\tau_1=2 \beta$ and $\tau_2>2(\alpha+1) \beta$. 
\end{proof}

Essentially, this means that any parameter-free algorithm that provides a multiplicative approximation for all values of $T$ is forced to use $\Theta(m)$ space, even on streams where the actual triangle count is moderately large. More generally, if the algorithm is only required to work when $T\geq \beta$, then $\Omega(m/\sqrt{\beta})$ space is still necessary. At a high level, before the algorithm can infer that $T$ is large, it must still be capable of distinguishing between $0$ and $\beta$ triangles, and that already requires $\Omega(m/\sqrt{\beta})$ space.

%% file: verified-algorithm-proof.tex
\section{Verified Triangle Counting}
\label{appendix:verifiedalgorithm}
\subsection{Algorithm}
Recall our notation: For an edge $e$ and a set of edges $F$, let $t_e^F$ be the number of triangles in $\{e\}\cup F$ that include $e$. Let $t_e:=t_e^E$ where $E$ is the edge set of the graph, $\Delta_e(G) = \Delta_e = \max_{e} t_e$. 


In this section, we prove \autoref{thm:bb}.
\verifiedalg*
During a single pass, the algorithm collects the various sets of edges.
The role of these is as follows: The set $P$ will likely include most of the heavy edges, i.e., edges involved in many  triangles. Sets  $E_0, E_1,\ldots, E_{\log \sqrt{\tguess}-1}$ will be used to construct $P$. Together $P$ and  $E_{\sqrt{\tguess}}$ will be used to estimate the number of triangles that include heavy edges. Another set $S$ will be used to estimate the number of triangles that don't share an edge with many other triangles, i.e., triangles consisting of three light edges.
\begin{itemize}
\item {\bf Finding Potentially Heavy Edges:} For $i=0, 1, \ldots, \log \sqrt{\tguess}$, let $V_i$ be a subset of vertices where each vertex is sampled with probability  $p_i:=\min\{1,10 \delta^{-1} \epsilon^{-2} (\log n) /2^i\}$. 
Let $E_i$ be defined as the set of edges incident to $V_i$ amongst the first $q_im$ elements of the stream, where $q_i:=2^i/\sqrt{\tguess}$. During a pass over the stream, store the following set of edges:
\[
P:=\{e\in E: \mbox{position of $e$ in stream is $> q_im$ and $t_e^{E_i}\geq 1$ for some $i$}  \}
\]

\item {\bf Colorful Sampling:} Independently assign one of $C:=\delta \epsilon^2 \sqrt{\tau}/7$ ``colors'' chosen uniformly at random  to each node of the graph. During a pass over the stream, collect all edges whose endpoints are monochromatic. Call this set  $S$. 
%

\item {\bf Post-Processing:} Let $O=E_{\log(\sqrt{\tguess})}$ and $p=p_{\log(\sqrt{\tguess})}$.
$$
\oracle(e)=
\begin{cases}
\light & \mbox{ if } t_e^{O}< p\sqrt{\tguess}\\
\heavy & \mbox{ otherwise}
\end{cases}
$$
The oracle will henceforth be used to \emph{define} whether an edge is heavy or light.\footnote{That is, while we will later show that $e$ being defined as heavy/light will roughly correspond to whether $t_e$ is larger/smaller than $\sqrt{\tguess}$, the actual definition is a function of the edges sampled by the algorithm; this will help significantly with the analysis. 
Note that the oracle is defined independently of the ordering of the data stream because $E_{\log(\sqrt{\tguess})}$ is constructed based on the entire stream rather than a strict prefix.} Let $E_\heavy=\{e\in E: \oracle(e)=\heavy\}$
and let $\tilde{T}_0$ be the number of triangles amongst $S\setminus E_\heavy$. Define
\[
X= \tilde{T}_0 C^2 
+
\frac{1}{p}\sum_{e\in P\cap E_\heavy} \left ( t_{e,0}^{O}+t_{e,1}^{O}/2+t_{e,2}^{O}/3\right )
\]
where $t_{e,i}^O$ is the number of triangles including $e$ where $i$ of the other two edges in $O$ are heavy.  The coefficients of $t_{e,i}^O$ take into account that a triangle with multiple heavy edges can be counted from the perspective of multiple edges and hence we need to compensate for overcounting. {If $X\geq \tau(1-\epsilon)$ return $X$, if not output \fail.}

\end{itemize}

%
 
 \subsection{Analysis} We first analyze the space used by the algorithm:

\begin{theorem}[Space Bound]\label{thm:finalfrontier}
The expected space used by the algorithm is $\tilde{O}(\epsilon^{-2}\delta^{-2} (m/\sqrt{\tau}+T/\tau))$.
\end{theorem}
\begin{proof}
The expected numbers of edges in all of the  $E_i$ sets is at most \[m\sum_{i=0}^{\log \sqrt{\tau }} 2p_i q_i = O(1)\cdot m \sum_{i=0}^{\log \sqrt{\tau}} 2^i/\sqrt{\tau} \cdot \delta^{-1} \epsilon^{-2} \log(n) /2^i\leq \tilde{O} (\delta^{-1} \epsilon^{-2} m/\sqrt{\tau} 
)
\ . \] 
The expected number of edges in $P$ is at most 
\begin{eqnarray*}
\sum_{i=0}^{\log \sqrt{\tau }} t_e p_i q_i^2 = 
O(\delta^{-1} \epsilon^{-2}  \log n) \cdot \sum_{i=0}^{\log \sqrt{\tau }} t_e 2^{2i}/\tau \cdot  1/2^i
&=& 
O(\delta^{-1} \epsilon^{-2} \log n) \cdot \sum_{i=0}^{\log \sqrt{\tau }} t_e 2^{i}/\tau \\
&=&  O(\delta^{-1} \epsilon^{-2} (\log n) T/\sqrt{\tau} )
 \ . \end{eqnarray*}
The expected number of edges in $S$ is $m/C=O(\epsilon^{-2} \delta^{-1} m/\sqrt{\tau})$. 
\end{proof}

To prove the accuracy guarantees we will argue that  a) the oracle distinguishes between edges $e$ with large or small $t_e$ values with sufficient accuracy, b) the algorithm stores almost all  edges with large $t_e$ values, and c) the number of triangles made of edges with small $t_e$ values, can be accurately estimated using the colorful sampling. 
 
We start by proving the oracle guarantees in a way that covers both \ppass~and the verified algorithm. We remind the reader of the relevant lemma.
\orcguarantees*
\begin{proof}
    Note that for every edge $e$ we have $t_e^O\sim \textup{Bin}(t_e,p)$, so $\Exp{t_E^O} = pt_e$. Suppose $t_e^O < p\gamma$ yet $t_e > 2\gamma$. Then $t_e^O<p\gamma$ implies $t_e^O < (1/2)t_ep$ and by Chernoff,
    \begin{align*}
    \prob{t_e^O < p\gamma}&\le \prob{t_e^O < (1/2)t_ep}\le \exp\rbk{-\frac{(1/2)^2}{2}t_ep}\\
    &\le \exp\rbk{-\frac{1}{4}p\gamma}\le \exp\rbk{-\frac{50\varepsilon^{-2}\log(n)}{4}}\le n^{-10}.
    \end{align*}
    On the other hand, suppose $\oracle(e) = \heavy$. We first show that $t_e$ must be large. Suppose $t_e < \gamma/2$, and therefore $\mu\defeq \Exp{t_e^O} \le p\gamma/2$. We write $p\gamma = (1+\kappa)\mu$ for a fitting $\kappa >1$. Therefore by Chernoff
    \begin{align*}
        \prob{t_e^O\ge p\gamma} &\le \prob{t_e^O\ge (1+\kappa)\mu}\le \exp\rbk{-\frac{\kappa \mu}{ 3}} = \exp\rbk{-\frac{p\gamma - \mu}{3}} \\
        &\le \exp\rbk{-\frac{p\gamma}{6}}\le \exp\rbk{-\frac{50\varepsilon^{-2}\log(n)}{6}} \le n^{-7}
    \end{align*}
    so with high probability, $t_e \ge \gamma/2$. The following follows from a standard two sided Chernoff bound conditioned on the above calculation:
    \begin{align*}
        \prob{|t_e^O-\mu|\ge\varepsilon\mu}\le 2\exp\rbk{-\frac{\varepsilon^{2}}{3}\mu }\le 2\exp\rbk{-\frac{50}{3}\log(n)} \le n^{-7}
    \end{align*}
    concluding the proof.
\end{proof}
\begin{lemma}[Oracle Guarantees] \label{lem:orcguarntees}
With high probability, $t_e^{O}=(1\pm \epsilon) t_e p$ for all heavy edges $e$ and $t_e\leq 2\sqrt{\tguess}$ for all light edges. 
\end{lemma}
\begin{proof}
Note that for any edge $e$, we have $t_e^{O}\sim \textup{Bin}(t_e,p)$ where  $p$ was defined   to be $p_{\log(\sqrt{\tguess})}=10 \delta^{-1} \epsilon^{-2} (\log n) /\sqrt{\tau}$.  The result follows from the Chernoff bound as in \ref{lem:oracleguarantees} and then taking the union bound over all edges.
\end{proof}
The next lemma establishes that $P$, the set of potentially heavy edges we construct, will contain most of the edges that ultimately will be defined to be heavy by the oracle. 

\begin{lemma}[Missing Heavy Edges]\label{missingmass}
$\sum_{e\in E^\heavy \setminus P} t_e \leq \epsilon T$ with probability at least $1-\delta$. 
\end{lemma}
\begin{proof}
First note that $|E_\heavy|\leq 4T/\sqrt{\tguess}$ since, by Lemma \ref{lem:orcguarntees}, if $e$ is heavy then $t_e\geq p\sqrt{\tguess}/(p(1+\epsilon))= \sqrt{\tguess}/(1+\epsilon)$ and summing $t_e$ over all heavy edges is at most $3T$. Next we consider the probability a heavy edge is not in $P$. To do this, let $i=\lceil \log_2 ( \tguess/(\delta^{-1}\epsilon^{-2} t_e) )\rceil$ 
and note that if $e$ does not appear within the first $q_i m$ edges of the stream then the probability $e$ is not added to $P$ is at most
\begin{eqnarray*}
(1-q_i^2 p_i)^{t_e}  \leq  e^{-t_e q_i^2 p_i} \leq e^{-t_e 10 \delta^{-1}\epsilon^{-2} (\log n) 2^{i}/\tguess}
 \leq  1/n^{10}
\end{eqnarray*}
because the events that different triangles are formed between $e$ and edges in $E_i$ are negatively associated. 

Therefore, the probability that we do not include heavy edge $e$ in $P$ is 
\begin{eqnarray*}
\prob{e\in E^\heavy\setminus P}
& \leq & 1/n^{10}+
\prob{e \mbox{ appears in a prefix of length $q_im$}} \\
& = &  1/n^{10}+2^{i}/\sqrt{\tguess}   \leq   2 \cdot \frac{2\delta \sqrt{\tguess}}{\epsilon^{-2} t_e} \ ,
\end{eqnarray*}
where we used the fact that $1/n^{10}$ was dominated by the second term for sufficiently large $n$.

Therefore:
\begin{eqnarray*}
\expec{\sum_{e\in E^\heavy \setminus P} t_e}
 \leq  \sum_{e\in E^\heavy} \frac{4\delta\sqrt{\tguess}}{\epsilon^{-2} t_e}\cdot t_e
 \leq  |E^\heavy| \cdot \frac{4\delta \sqrt{\tguess}}{\epsilon^{-2}}
\leq \frac{16\delta T}{\epsilon^{-2}}
\end{eqnarray*}
using the fact  $|E^\heavy| \leq 4T/\sqrt{\tguess}$. The result then follows by an application of the Markov bound assuming $\epsilon<1/16$.
\end{proof}

The total number of triangles can be written as:
\[T= T_0+\sum_{e\in E_\heavy}(t_{e,0}+t_{e,1}/2+t_{e,2}/3)\]
where $T_0$ is the number of triangles with $0$ heavy edges and $T_\heavy$ is the number of triangles with at least 1 heavy edge. $t_{e,i}$ is the number of triangles including $e$ where exactly $i$ of the other edges are heavy. The next two lemmas consider errors incurred when estimating the terms $T_0$ and $T_1+T_2+T_3$.

\begin{lemma}[Light Triangles]\label{lightcontrib} 
$\prob{|\tilde{T}_0C^2 - T_0| \geq  \epsilon \max(T,\tau/2) }\leq \delta$. 
\end{lemma}
\begin{proof}

Since $C\leq \tau$, we can bound the variance as:
\[\var{\tilde{T}_0C^2 } \leq T_0 C^2 + 3T_0 C\cdot (\eta_\light-1) 
\leq T_0 C^2 + 6T_0\sqrt{\tau} C\leq 7C T_0\sqrt{\tau}\]
where $\eta_\light\leq 2\sqrt{\tau}$ is the maximum number of triangles amongst $E\setminus E_\heavy$ sharing the same edge. 
Hence, by an application of Chebyshev bound
\[
\prob{|\tilde{T}_0C^2 - T_0| \geq  \epsilon \max(T,\tau/2) }\leq \frac{7C T_0\sqrt{\tau}}{\epsilon^2 \max(T,\tau/2)^2}\]
If $T\geq \tau/2$, this is at most \[7C\sqrt{\tau} \epsilon^{-2}T^{-1}\leq 14 C \epsilon^{-2}/\sqrt{\tau} \  \] and if $T\leq \tau/2$, this is also at most \[7C\sqrt{\tau} (\tau /2) \epsilon^{-2}(\tau/2)^{-2}\leq 14 C \epsilon^{-2}/\sqrt{\tau} \ . \]
Hence, setting $C=\delta \epsilon^{2} \sqrt{\tau}/14$ suffices to ensure the probability is at most $\delta$.
\end{proof}

\begin{lemma}[Heavy Triangles]\label{heavycontrib}
With probability at least $1-1/\poly(n)$,
\[\frac{1}{p} \sum_{e\in E_\heavy} \left (t_{e,0}^{O}+t_{e,1}^{O}/2+t_{e,2}^{O}/3 \right ) = (1\pm \epsilon) (T_1+T_2+T_3)
\ . \]
\end{lemma}
\begin{proof}
By an application of the Chernoff bound,
for each heavy edge $e$, with high probability 
\[t_{e,0}^{O}+t_{e,1}^{O}/2+t_{e,2}^{O}/3=(1\pm \epsilon)p(t_{e,0}+t_{e,1}/2+t_{e,2}/3) \ . \]
The result follows by summing over the heavy edges.
\end{proof}

We conclude with the proof of Theorem \ref{thm:bb}.
\begin{proof}[Proof of Theorem \ref{thm:bb}]
If $ \tau/2 \leq T$, by combining Lemmas  \ref{lem:orcguarntees}, \ref{missingmass}, \ref{lightcontrib}, and \ref{heavycontrib}, then with probability at least $1-2\delta-1/\poly(n)$, we have: \[(1-2\epsilon)T\leq 
(T_0-\epsilon T) +(1-\epsilon) (T_1+T_2+T_3-\epsilon T) \leq X\leq (T_0+\epsilon T)+(1+\epsilon) (T_1+T_2+T_3)\leq (1+2\epsilon) T  \ .\] Note that if $\tau \leq T$ then the output is $(1\pm 2\epsilon)T$ and in particular is at least $(1-2\epsilon) T\geq (1-2\epsilon)\tau$.

If $ \tau/2 \geq T$, by combining Lemmas \ref{lem:orcguarntees}, \ref{lightcontrib}, and \ref{heavycontrib}, then with probability at least $1-2\delta-1/\poly(n)$ 
 \[
X\leq T_0+\epsilon \tau/2 +(1+\epsilon)(T_1+T_2+T_3)\leq (1+\epsilon) T+ \epsilon \tau /2
\leq  (1+\epsilon) \tau /2+ \epsilon \tau /2< (1-2\epsilon) \tau \ . 
\]
assuming $\epsilon \leq 1/6$.
Re-parameterizing $\epsilon \leftarrow \epsilon/2$ and $\delta\leftarrow \delta/2$ in the algorithm yields the required accuracy results. The space bound is proved in Theorem \ref{thm:finalfrontier}.
\end{proof}